\documentclass[12pt]{article}
\usepackage[T1]{fontenc}
\usepackage{kpfonts}
\usepackage[dvipsnames]{xcolor}
\usepackage{pdfpages}
\usepackage{sgame}
\usepackage{accents}
\usepackage{tikz-cd}
\usepackage{float}
\usepackage[round]{natbib}   
\usepackage{multirow}
\usepackage{dutchcal}
\usepackage{pdfpages}
\usepackage{sgame}
\usepackage{mathtools}
\usepackage{amsthm}
\usepackage{enumitem}
\usepackage[margin=1in]{geometry}
\usepackage{caption}
\usepackage{subcaption}
\usepackage{epigraph}
\usetikzlibrary{calc}
\usetikzlibrary{shapes,arrows}

\newcommand{\ubar}[1]{\underaccent{\bar}{#1}}

\newtheorem{theorem}{Theorem}[section]
\newtheorem{proposition}[theorem]{Proposition}
\newtheorem{lemma}[theorem]{Lemma}
\newtheorem{corollary}[theorem]{Corollary}
\newtheorem{claim}[theorem]{Claim}

\theoremstyle{definition}

\newtheorem{remark}[theorem]{Remark}

\usepackage{hyperref}
\hypersetup{
    colorlinks=true,
    linkcolor=Violet,
    filecolor=Violet,      
    urlcolor=Violet,
    citecolor=Violet,
}

\providecommand{\keywords}[1]{\textbf{\textit{Keywords:}} #1}
\providecommand{\jel}[1]{\textbf{\textit{JEL Classifications:}} #1}

\DeclareMathOperator{\sgn}{sgn}

\begin{document}
\author{Vasudha Jain\thanks{University of Texas at Austin} \and Mark Whitmeyer\thanks{Hausdorff Center for Mathematics \& Institute for Microeconomics, University of Bonn \newline Email: \href{mailto:mark.whitmeyer@gmail.com}{mark.whitmeyer@gmail.com}. \newline We thank Francesc Dilm\'e, Joseph Whitmeyer, Thomas Wiseman, and a seminar audience at the University of Bonn for helpful comments and feedback. The second author's work was funded by the DFG under Germany's Excellence Strategy-GZ 2047/1, Projekt-ID 390685813.} }

\title{Search and Competition with Flexible Investigations}

\date{\today}

\maketitle

\begin{abstract} 
We modify the standard model of price competition with horizontally differentiated products, imperfect information, and search frictions by allowing consumers to flexibly acquire information about a product's match value during their visits. We characterize a consumer's optimal search and information acquisition protocol and analyze the pricing game between firms. Notably, we establish that in search markets there are fundamental differences between search frictions and information frictions, which affect market prices, profits, and consumer welfare in markedly different ways. Although higher search costs beget higher prices (and profits for firms), higher information acquisition costs lead to \textit{lower} prices and may benefit consumers. We discuss implications of our findings for policies concerning disclosure rules and hidden fees. 
\end{abstract}
\\
\jel{C72; D82; D83} 
\newpage
\section{Introduction}

There is a large and rich literature on sequential consumer search. Many papers in this area share a similar common structure: there are collections of possibly heterogeneous firms and consumers, the latter of which wish to purchase products from the firms. These consumers (or at least some subset of them) visit firms in sequence in order to discover important information about the products including, e.g., their qualities, if the products are vertically differentiated; match values, if there is horizontal differentiation; and/or prices. The essential decisions made by consumers are whom to visit (if search is directed), when to stop, and whether to purchase upon stopping (and from whom).

Despite the considerable control they possess regarding their search behavior, consumers in these models are surprising passive in a different sense. Namely, when a consumer visits a firm she is unable to affect the information her visit reveals--there is an equivalence between visiting a firm and acquiring information about its product(s). She observes (at least some of) the pertinent information (prices, qualities, and/or match values) or the draw of some signal correlated with these values, but may not control this information \textit{directly}.

In many settings this is unrealistic. Consider, for instance, a prospective homeowner looking for a house. She visits homes on the market in sequence; but rather than absorb information about a house effortlessly upon visiting it, she must acquire this information actively. Not only does the consumer face a sequential search problem as she moves from house to house, but within each visit, she has a dynamic information acquisition problem. She chooses which rooms to visit and how long to linger in each; what questions to ask the real-estate agent; how thoroughly to inspect the plumbing, lighting, or paint condition; and how many dark corners to scrutinize for mold.\footnote{In fact there are a number of guides online that advise consumers how to behave optimally during such visits; one urges its readers, ``Once inside a home, try everything. Follow common courtesy but don't be shy.'' (\url{https://www.redfin.com/home-buying-guide/what-to-look-for})}

This phenomenon is not limited merely to the housing market. Indeed, it is difficult to think of many markets in which such information acquisition does not take place. A prospective car buyer chooses how many features of a car to inspect when at a dealer, a consumer shopping for groceries chooses how much of each product's label to peruse, a teenager looking for a prom dress scrutinizes candidates in the many mirrors of a dress store's changing room, and a music enthusiast looking for new headphones plays various songs to test out different facets of a device's sound quality and feel.

Importantly, this information acquisition process is dynamic. If a prospective homeowner sees that a house's roof is caving in, she stops her tour right away and moves on to the next property. Likewise, a young businesswoman may not test drive a mini-van in garish pink--she has already moved on to the sleek sedan in dark grey. Information is costly and so consumers investigate until they have ``seen enough,'' at which point they decide whether to stop or move on to the next firm or product.

Our goal in this paper is to explore the effect of this flexible information acquisition on consumer and firm behavior in an otherwise standard sequential search market. There are a large number of consumers with unit demand for a horizontally differentiated product and a large mass of \textit{ex ante} indistinguishable firms that sell the product. It costs a consumer $c > 0$ to visit a firm; and when at a firm she acquires information about the product's idiosyncratic fit by sampling various ``attributes''--which inform her valuation--in sequence, at a cost of $\gamma > 0$ per attribute.

The firms compete by setting prices, which--following much of the literature--can only be discovered by the consumer upon visiting a firm. In our main analysis we assume that when visiting a firm, a consumer observes the firm's price before acquiring information (before she starts sampling attributes at the firm). In many circumstances this is realistic: most retailers of consumer goods post prices that are non-negotiable and which consumers note before inspecting the items. Due to the symmetry of the model we focus on symmetric equilibria.

Within this setting, we tackle a number of questions. First, it is a standard result that in search models with hidden prices, an increase in the search cost leads to an increase in prices. Does this hold when there are information frictions as well? What about the effect of information frictions on prices? That is, do greater information frictions lead to higher prices? Do laws mandating maximal disclosure always benefit consumers? Do they always hurt firms?


Importantly, a consumer's optimal information acquisition protocol is shaped by the prices set by firms. Consequently, this gives her considerable power to react to a price deviation by a firm--even though she can not observe a price change before her visit, such a deviation will affect her optimal information acquisition protocol at that firm. This is realistic--consumers, upon seeing an unexpectedly high price, may just decide that the product is too expensive to consider at all.

We find that (explicit) search costs continue to have the same effect as in models without flexible learning. That is, prices and firm profits increase in (explicit) search costs, whereas consumer welfare decreases. On the other hand, the effect of information frictions is unexpectedly ambiguous. Prices decrease in the level of information frictions, and so consumer welfare may rise as a result. In fact, for certain regions of the parameter universe, an increase in information frictions leads to a Pareto increase in welfare. The expected duration of search is also non-monotone in the amount of information frictions. 

One notable implication of our results is that a firm's market power is generated entirely by the search costs (think switching costs) and not by the information (learning) costs. To elaborate, it is only the existence of this search cost that prevents a consumer from departing the firm at the first sign of bad news. The information frictions, in turn, generate a downward sloping demand curve for firms, which becomes steeper (\textit{ceteris paribus}) as information frictions increase. Thus, an increase in information frictions results in a decrease in the equilibrium price, and this decrease may be so great as to outweigh a consumer's direct welfare loss as a result of increased frictions. Firms are unaffected by the change in information frictions, since their profits rely on the market power granted to them by the search cost.


Later on we explore a natural modification of the model: we stipulate that a consumer does not observe a firm's price until \textit{after} she has acquired information. This leads to an ``informational hold-up problem,'' and in the unique equilibrium, firms mix over prices. Thus, our model generates price dispersion even when prices are not posted. Surprisingly, the qualitative results from the main setting continue to hold: explicit search costs benefit firms at the expense of consumers but information frictions may be Pareto improving.

We also compare the results from the two timing specifications. Consumers strictly prefer that prices be observable before learning. Firms, on the other hand, may prefer that prices be hidden, but only if the expected quality of the product and the level of information frictions are sufficiently high. Otherwise, firms prefer that prices be observable as well. The comparisons in this section allow us to speak to the recent debate about hidden fees, drip pricing, and other varieties of hidden prices. Our results suggest that laws that mandate transparent prices do indeed benefit consumers, since they allow them to learn efficiently. Moreover, firms may benefit from such laws as well.

We finish this section by discussing the relevant papers in the literature. Section \ref{model} lays out the model and establishes some preliminary results. Section \ref{search} characterizes a consumer's optimal search and information acquisition protocol, Section \ref{monomono} briefly describes a monopolist's pricing behavior, and Section \ref{infinite} explores the equilibrium in the pricing game between the firms. In Section \ref{hidden} we modify the game so that prices are now unobservable to the consumer before she acquires information and we look at the pricing game between the firms. Section \ref{fees} compares welfare in the two timing regimes and Section \ref{discus} concludes. All proofs are left to the appendix,\footnote{Due to the algebra-heavy nature of some of the comparative statics derivations, we also have a \href{https://whitmeyerhome.files.wordpress.com/2021/04/search_and_competition_supplementary_appendix.pdf}{Supplementary Appendix} that provides Mathematica code to corroborate our calculations.} unless otherwise noted. 

\subsection{Related Work}

The paper closest to this one is \cite*{branco12},\footnote{We discuss their model in detail in the next section, when we describe our setup. Henceforth, we refer to them as \hypertarget{BSV}{BSV}.} who look at the monopolist's problem when consumers flexibly acquire information after observing the monopolist's price. As in our search setting, they find that information frictions may Pareto-improve welfare. The portion of our paper in which prices are not observable before consumers acquire information is closely related to \cite*{ravid2020learning}, who explore a bilateral-trade setting in which a consumer acquires costly information about a monopolist's product before observing the seller's price. Their focus is on the limiting case as information frictions vanish, and they establish that in the limit, the equilibrium converges to the Pareto-worst free learning equilibrium.\footnote{A consumer's learning problem in our setting is not quite the same as that in \cite*{ravid2020learning}. In their model the consumer faces a static problem: she chooses a signal about the value of a product, about which she has some prior distribution. As shown by \cite*{wald}, when the prior is binary there is an equivalence between the static information acquisition problem, in which a decision-maker chooses a distribution over posteriors subject to a posterior-separable cost, and a dynamic Wald stopping problem. Beyond the binary-prior setting, this equivalence vanishes. Moreover, our problem is not a Wald problem--a consumer directly observes various components that affect her expected value for the product.} 

Other papers that look at a monopolist selling to consumers who may acquire information flexibly include \cite*{branco2016too}, \cite*{pease2018shopping}, and \cite*{lang2019try}. In the latter two papers, a consumer observes a monopolist's posted price before observing a diffusion process correlated with the product's quality at a flow cost. \cite*{lang2019try} assumes that the consumer's prior valuation for the product is distributed normally and finds that in regions of the parameters with active information acquisition, increased information frictions improve the monopolist's profit if and only if the \textit{ex ante} expected value for the product is sufficiently high. \cite*{pease2018shopping} assumes a binary value for the product, and also finds that the level of information frictions has an ambiguous effect on firm profit. \cite*{branco2016too} modify the setting of \hyperlink{BSV}{BSV} by allowing the seller to choose how much information to provide about the product. They find that the seller may want to induce information overload and provide more information than is optimal for the consumer.

\cite*{ke2016search} alter \hyperlink{BSV}{BSV} by looking at the information acquisition problem of a consumer inspecting the attributes of multiple products in continuous time. At an instant, the consumer may shift her attention from one product to another and acquires information at each in the same fashion as \hyperlink{BSV}{BSV} (and our paper)--by observing Brownian motion subject to a constant flow cost. Given the difficulty of the consumer's problem (it is a PDE problem with ambiguous boundary conditions), the focus of their paper is on the optimal information acquisition protocol of a consumer when there are just two products. They also provide numerical analysis of the pricing problem faced by a multi-product monopolist.

To our knowledge, the only other work to explore the dual effects of information and search frictions in a sequential search model is \cite{guo2021endogenous}. In his model, a consumer searches sequentially and when at a firm draws her value from a family of rotation ordered distributions. The primary focus of his paper is the case in which a consumer chooses from a collection of uniform distributions. There, he finds that the explicit search cost and the cost of informativeness (in the rotation order) of a consumer's posterior value distribution operate in similar ways: equilibrium price and profit increase in both types of friction but may drop discontinuously for high frictions (of either variety). The reverse happens for consumer welfare. Our setup is different in that our consumers tackle a sequential information acquisition problem--one that is equivalent to a static problem in which they choose any distribution over posterior values whose mean is the prior, subject to a cost. This allows us to illustrate how search and information frictions operate in fundamentally different ways and have dramatically different effects on equilibrium behavior and outcomes.

When a consumer does not observe a firm's price before acquiring information about its product, she faces an informational hold-up problem. Similar to \cite*{ravid2020learning}, \cite*{Roes} investigate the consumer-optimal signal in a bilateral trade environment, which is subsequently extended by \cite*{yang}, who allows the seller to have a more general production technology. In the same setting, \cite*{Condor} allow the buyer to choose her distribution over values for the product. \cite*{hu}, in turn, explore the buyer-optimal signal in the random sequential search setting of \cite*{asher}. In each paper, some form of truncated Pareto distribution is optimal (which is intuitive, since it is those distributions over values/posteriors that generate unit-elastic demand).

The literature that explores more standard versions of the hold-up problem in bilateral trade is also relevant. \cite*{bar2012information} look at the scenario where a seller can privately invest in the quality along one dimension of a multi-characteristic good, after which consumers may (at a cost) acquire information about one dimension of the product. The cost of information affects the firm's investment via its influence on the consumers' behavior and may lower consumer welfare and the firm's profit. In contrast,  there is no investment in our model and so the cost of information affects firms' behavior only through its influence on their prices.

More recently, \cite*{dilme2019pre} looks at a scenario in which a seller can privately invest in the quality of its product before a buyer makes a take-it-or-leave-it offer. As we find in the hidden prices portion of this paper, and following similar logic, Dilm\'e shows that the seller randomizes over investment levels. He also finds that modifying the information structure of the game--so that the buyer observes the mean of the seller's investment policy--aids the seller, who increases his investment level and his profit at the expense of the buyer. Along these lines, \cite*{lau2008information}, \cite*{hermalin2009information}, and \cite*{nguyen2019information} all look at various facets of information design in the hold-up problem. \cite*{rao2021search} embeds a similar hold-up problem into a labor search model, in which workers invest in human capital. With hidden investment, Rao shows that the equilibrium exhibits both wage and skill dispersion akin to our information and price dispersion. 

Naturally, our paper is also related to the broader collection of works that look at the effects of consumer rational inattention on market behavior. Of particular pertinence is \cite*{matvejka2012simple}, who explore a model in which a consumer must incur a cost to evaluate the offers of firms within a market. Notably, they show that prices are increasing (and, hence, consumer welfare decreasing) in the size of the information frictions. Also related is \cite*{dukes1}, who look at a simultaneous search model, in which a consumer must evaluate products in an oligopolistic market subject to an evaluation cost. \textit{ex ante} she chooses how many firms about which to acquire information and to which depth. Crucially, the dynamic aspect of our problem--in which a consumer's optimal behavior (both regarding her search behavior and her information acquisition protocol) is determined each period--allows a consumer to react to price deviations by firms. They find that an increase in information costs may benefit consumers, though at the expense of total welfare. Other studies involve incentives for firms to ``stay under consumers' radar'' (\cite*{de2014competing}) or ``get on consumers' radar'' (\cite*{bordalo2016competition}), which forces are absent from our paper.

\section{The Model}\label{model}

Our model extends the framework of \hyperlink{BSV}{BSV} to a random search setting. There is a continuum (unit mass) of selling firms and a continuum (unit mass) of consumers. Each firm's product has $T$ different attributes that are \textit{ex ante} unknown. The value of each attribute to a consumer is a random variable. For simplicity, these random variables are i.i.d. and have mean $0$. More specifically, the value of attribute $j$ to a consumer is the random variable $X_{j}$, which takes values $-z$ and $z$ with equal probability, where $z > 0$.

If firm $i$ sets price $p_{i}$ then a consumer's realized utility from product $i$ is $\mu - p_{i} + \sum_{j=1}^{T}X_{j}$, where $\mu \in \mathbb{R}$ is the consumer's prior value for the product. We interpret $\mu$ to be a consumer's baseline value for the products for sale, and assume that it is the same for each firm. Given this, it is the idiosyncratic attributes of a firm's product that set it apart from its competitors.
For example, suppose the consumer is shopping for a bicycle. Given her habits and needs she has some underlying valuation $\mu$ for a generic bike. Then, once she arrives at a bike store, she may test out the various attributes of a bike, each of which shapes her value for that specific bicycle further--she likes its color but not the shape of its seat, she finds that the handlebars are angled too far forward but likes their grip, and so on. Crucially, these various attributes do not affect her valuation for other bikes.

Following \hyperlink{BSV}{BSV}, we assume that the number of attributes, $T$, is infinity. As they note, this has the reasonable interpretation that the set of attributes can never be ``fully searched.'' Naturally, this assumption also helps with the problem's tractability.

We make a distinction between search and information acquisition. Each period a consumer decides whether to stop and purchase from the current firm, select her outside option, or continue her search, in which case she is randomly matched with another firm in the market. When a consumer visits a firm, she acquires information flexibly. As is routine in the literature, a consumer's first period visit is free, but any subsequent search imposes on her some cost $c > 0$ per period. 

When a consumer visits a firm, she acquires information flexibly. She does this by examining attributes one-by-one, in sequence, at a cost of $\gamma > 0$ per attribute. Examining attribute $j$ reveals the realization of random variable $X_{j}$, and after examining $t$ attributes, a consumer's expected valuation from purchasing firm $i$'s product (so, net of that firm's price) is
\[u_{i}(t) = v_{i} + \sum_{j=1}^{t}x_{j} + \sum_{j=t+1}^{T}\mathbb{E}\left[X_{j}|t\right] =  v_{i} + \sum_{j=1}^{t}x_{j} \text{ ,}\]
where $v_{i} \coloneqq \mu - p_{i}$; and since, by assumption, $\mathbb{E}\left[X_{j}|t\right] = 0$ for all $j, t$.

We continue to follow \hyperlink{BSV}{BSV} by simplifying the information acquisition problem faced by a consumer at firm $i$. We assume that the change in utility from each attribute, $z$, becomes infinitesimal as $T$ grows infinitely large.  This captures the fact that the importance of any one attribute is negligible compared to the value of the product on the whole. In the limit, the $u_{i}(t)$ process becomes a Brownian motion--cropping the argument, $du = \sigma dW$, where $W$ is a standardized Brownian motion--with standard deviation $\sigma$, where $z = \sigma \sqrt{dt}$.\footnote{As noted by \cite*{lang2019try}, there is a subtle technical issue to this approach in that the ``actual value for the product $\lim_{t \to \infty}\int_{0}^{t}W_sds$ is ill-defined.'' However, we argue that this modeling choice is akin to the use of the improper uniform distribution in the global games literature, and which greatly helps us with simplicity at the cost of slight impropriety. Furthermore, for most of this paper, we use a convenient reformulation of the consumer's stopping problem as a static information acquisition problem, and we urge any objectors to the dynamic set-up to view it exclusively as the static problem instead.} When a consumer is at a firm, each instant she decides whether to keep observing attributes or to stop doing so. Once she stops acquiring information at a firm, she either purchases from the current firm, selects her outside option, or continues her search. Throughout, we assume that $c < \sigma^2/\left(4\gamma\right)$. This assumption--that search costs are not too high--is a necessary condition for the existence of equilibria with active search. For simplicity, we normalize the value of the consumer's outside option to $0$.

The timing of the game is summarized in Figure \ref{timing} below.

\tikzstyle{block} = [rectangle, draw, 
    text width=15em, text centered, rounded corners]
    
\tikzstyle{line} = [draw, -latex']

\begin{figure}[ht]
\centering
\begin{tikzpicture}[node distance = 8cm, auto, scale=0.6, transform shape]
    \node [block] (0) {Firms simultaneously (and privately) set prices.};
    \node [block, right of=0] (1) {Consumer is randomly matched with a firm whose price she then observes.};
    \node [block, right of=1] (2) {Consumer decides when to stop costly acquisition of information about that firm's attributes.};
    \node [block,  above right of=2] (3) {Consumer pays $c$ to be randomly matched with another firm, and the same process is repeated.};
    \node [block, below right of=2] (4) {Consumer buys from the current firm or selects her outside option and payoffs are realized.};

    \path  [line] (0) -- (1) ; 
     \path  [line] (1) -- (2) ; 
    \path [line] (2) --node[above, sloped]{Continues search} ++ (3);
    \path [line] (2) --node[below, sloped]{Stops search} ++ (4);
    
\end{tikzpicture}
\caption{Timing}
\label{timing}
\end{figure}

\subsection{Simplifying the Consumer's Information Acquisition Problem}
Next, we observe that the consumer's information acquisition problem can be reformulated from a dynamic optimal stopping problem to a much simpler static problem. Evidently, at a firm, a consumer's information acquisition strategy is a stopping time $\tau$, and any such $\tau$ generates a distribution over (expected values) $u$, which we denote
$F_{\tau}$. That is,
\[F_{\tau}\left(u\right) \coloneqq \mathbb{P}\left(U_{\tau} \leq u\right) \text{ .}\]
We denote by $\mathcal{M}\left(\mu\right)$ the set of (Borel) probability measures on $\mathbb{R}$ with mean $\mu$ and define the \textit{ex ante} cost of distribution $Q \in \mathcal{M}\left(\mu\right)$ over values to be the minimal cost at which the consumer can generate $Q$:
\[C\left(Q\right) \coloneqq \inf_{\tau: F_{\tau} = Q}\mathbb{E}\left[\gamma\tau\right] \text{ .}\]
Then, the results of, e.g., \cite*{church} imply the following remark:\footnote{This remark is a trivial modification of the main result of \cite*{root} (Theorem 2.1). See also \cite*{georgiadis} for a recent use of this result in an economic setting. This is an offshoot of the Skorokhod embedding problem in mathematics.}
\begin{remark}\label{equiv}
The \textit{ex ante} cost of distribution $Q \in \mathcal{M}\left(\mu\right)$ over values, $C: \mathcal{M}\left(\mu\right) \to \mathbb{R}^{+}$, is
\[C\left(Q\right) = \kappa \int_{-\infty}^{\infty}\left(x-\mu\right)^2dQ\left(x\right), \quad \text{where} \quad \kappa \coloneqq \gamma/\sigma^2 \text{ .}\]
\end{remark}
The precise formula follows via Ito's lemma: for value $u_{\tau}$,
\[c\left(u_{\tau}\right) = c\left(\mu\right) + \int_{0}^{\tau}\sigma c'\left(u_{t}\right)dW_{t} + \frac{1}{2}\int_{0}^{\tau}\sigma^2 c''\left(u_{t}\right)dt \text{ ,}\]
where $c\left(u\right) \coloneqq \kappa \left(u-\mu\right)^2$. Because $\tau$ is a.s. finite, taking expectations, the first term on the right side of the equal sign is $0$, and the middle term, being a martingale, vanishes. 

Accordingly, instead of solving a consumer's dynamic information acquisition problem at each firm, we may instead solve the following static problem. Let $g\left(u\right)$ be a consumer's payoff from stopping her information acquisition at value $u$. She solves
\[\max_{Q \in \mathcal{M}\left(\mu\right)}\left\{\mathbb{E}_{Q}\left[g(u)\right] - C\left(Q\right)\right\} = \max_{Q \in \mathcal{M}\left(\mu\right)}\left\{\int_{-\infty}^{\infty}g(u)dQ(u) - \kappa \int_{-\infty}^{\infty}\left(u-\mu\right)^2dQ\left(u\right)\right\} \text{ .}\]
\section{A Consumer's Behavior}\label{search}

Our first step in characterizing a consumer's optimal search and information acquisition protocol is to determine how she optimally acquires information at a seller $i$ charging price $p_{i}$ when she has some outside option of $a \geq 0$ ``in hand.''\footnote{$a$ is the maximum of her maximal continuation value and her explicit outside option (which is given exogenously).} Her payoff from obtaining value $u$ is
\[W\left(u\right) = \begin{cases}
a - \kappa\left(u-\mu\right)^{2}, \quad &u - p_{i} \leq a\\
u - p_{i} - \kappa\left(u-\mu\right)^{2}, \quad &u - p_{i} \geq a\\
\end{cases} \text{ .}\] 
Her optimal learning is given by the concavification\footnote{This approach is famously used by \cite*{aumann1995repeated} to characterize the limiting value of repeated games with one-sided incomplete information. \cite{kam} provide a recent promulgation of this technique.} of this function and is pinned down by the two posterior beliefs $u_{L}$ and $u_{H}$. Explicitly, \[\tag{$1$}\label{eq2}u_{L} = a + p_{i} - \frac{1}{4 \kappa}, \quad \text{and} \quad u_{H} = a + p_{i} + \frac{1}{4 \kappa}, \quad \text{where} \quad \mathbb{P}\left(u_{H}\right) = \frac{1}{2} + 2\kappa\left(\mu - p_{i}-a\right)\]
An example of $W$ and its concavification are depicted in Figure \ref{fig1}.
\begin{figure}
    \centering
    \includegraphics[scale=.3]{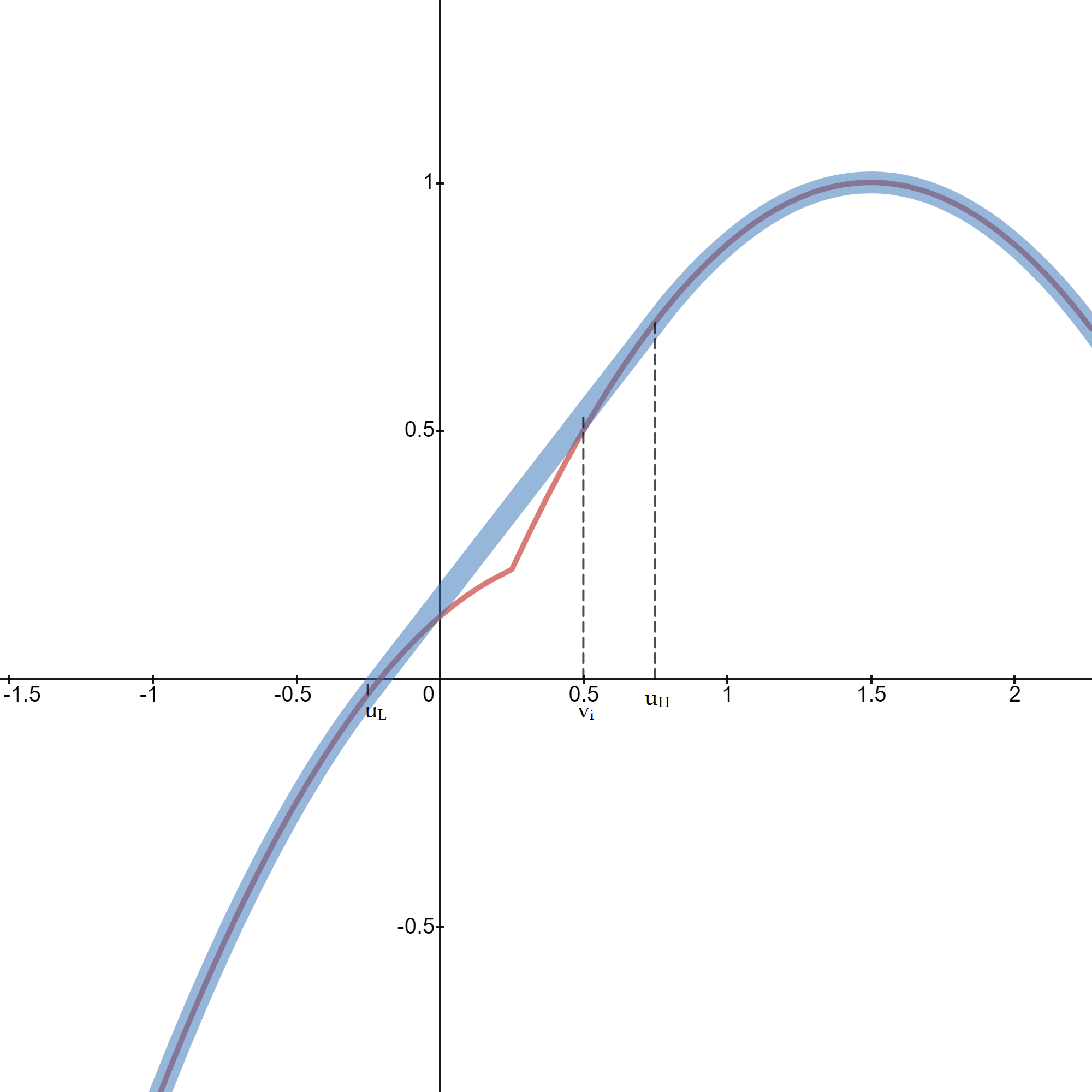}
    \caption{$W$ (in red) and its concavification (in blue) for $a = 1/4$, and $v_{i} \coloneqq \mu - p_{i} = \kappa = 1/2$.}
    \label{fig1}
\end{figure}
A consumer, given prior value $\mu$, price $p_{i}$, with $a \geq 0$ in hand, and behaving optimally, obtains the following payoff:
\[\tag{$2$}\label{eq3}W^{*}\left(\mu, p_{i}; a\right) = \begin{cases}
a, \quad &\mu \leq a + p_{i} - \frac{1}{4\kappa}\\
\frac{1}{16 \kappa} + \frac{\mu - p_{i} + a}{2} + \kappa\left(\mu - p_{i} - a\right)^{2}, \quad &a + p_{i} - \frac{1}{4\kappa} \leq \mu \leq a + p_{i} + \frac{1}{4\kappa}\\
\mu - p_{i}, \quad &a + p_{i} + \frac{1}{4\kappa} \leq \mu\\
\end{cases} \text{ .}\]
If the prior value, $\mu$, is in an interior region, a consumer learns and either becomes very optimistic about firm $i$'s product, in which case she purchases it; or becomes pessimistic about firm $i$'s product and leaves without purchase.  Equivalently, for an interior $\mu$, a consumer observes process $W$ until its value hits either $u_{L}$ or $u_{H}$ at which point she stops acquiring information and either purchases the product or continues her search elsewhere. Outside of the interior region, a consumer does not learn and either purchases from firm $i$ immediately or pursues other options. Predictably, as  $\kappa$ rises, learning occurs over a smaller range of prior values. Furthermore, for higher values of $a+p_i$, an immediate purchase from firm $i$ is optimal for a smaller range of priors, while pursuing other options without learning about firm $i$ is optimal for a larger range of priors.

\subsection{A Consumer's Search Protocol}

Because of our focus on symmetric equilibria, a consumer faces a stationary environment. It is easy to see that a consumer's optimal search protocol can be calculated recursively. Given this, suppose that the parameters are such that there exists active search and information acquisition (we will determine these precise values shortly). As a result, we define $\Phi$ to be a consumer's payoff under her optimal search and information acquisition protocol. Given that we have assumed there to be active search and information acquisition, the consumer's optimal information acquisition protocol is stipulated by Expression \ref{eq2}, in which we set $a = \Phi-c$. Moreover, using Expression \ref{eq3} (and again setting $a = \Phi-c$) we can calculate $\Phi$:
\[\Phi = \frac{1}{16 \kappa} + \frac{\mu - p + \Phi - c}{2} + \kappa\left(\Phi - c - \left(\mu - p\right)\right)^{2} \text{ ,}\]
where we cast aside the subscript $i$ due to symmetry. Rearranging this equation, we obtain
\[\Phi_{l} = \mu - p + c+\frac{1}{4\kappa} - \sqrt{\frac{c}{\kappa}} \text{ ,}\]
where the subscript denotes ``low price,'' introduced so as to distinguish this value from the maximal payoffs in other parameter regions. 

Now, let us tackle our (implicit) parametric assumptions. The first is that $\Phi_{l} - c\geq 0$, since the consumer must (weakly) prefer continuing her search over stopping to take her outside option. This holds if and only if
\[\tag{$3$}\label{ineq5}\mu - p \geq \sqrt{\frac{c}{\kappa}} - \frac{1}{4\kappa} \text{ .}\]
The second is that the learning is feasible; namely, that \[\Phi_{l} -c - \frac{1}{4\kappa} \leq \mu - p \leq \Phi_{l} -c + \frac{1}{4\kappa} \text{ .}\] Since $c > 0$, the left-hand inequality must always hold, but the right-hand one is satisfied if and only if our assumption $c \leq 1/\left(4\kappa\right)$ holds.

If Inequality \ref{ineq5} does not hold, a consumer's optimal learning is again given by Expression \ref{eq2}, in which $a = 0$; and we find that this protocol is optimal provided $-1/\left(4\kappa\right) \leq \mu - p \leq \sqrt{c/\kappa} - 1/\left(4\kappa\right)$. The consumer's payoff is \[\Phi_{m} = \frac{1}{16\kappa} + \frac{\mu - p}{2} + \kappa \left(\mu-p\right)^2 \text{ ,}\]
where the subscript stands for ``medium price.'' If $\mu - p \leq -1/\left(4\kappa\right)$, consumers neither search nor acquire information and there are no transactions. A consumer's payoff is merely $0$.

Compiling these observations, we have
\begin{proposition}
If $\mu - p \leq - 1/\left(4\kappa\right)$, there is no trade, and the consumer's payoff is $0$. If $-1/\left(4\kappa\right) \leq \mu - p \leq \sqrt{c/\kappa} - 1/\left(4\kappa\right)$, a consumer learns at one firm and obtains $\Phi_{m}$; and if $\mu-p \geq \sqrt{c/\kappa} - 1/\left(4\kappa\right)$, a consumer both learns and searches and obtains $\Phi_{l}$.
\end{proposition}

\section{A Monopolist's Problem}\label{monomono}

With a view towards our pricing equilibrium, let us briefly characterize a monopolist's optimal pricing policy in this environment. This is precisely the problem studied in \hyperlink{BSV}{BSV}, so this section's results merely repeat theirs (with the trivial generalization of allowing the the consumer to have an outside option of $a \geq 0$).

The monopolist's profit as a function of its price $p_{m}$ is
\[\Pi_{M}\left(p\right) = \begin{cases}
0, \quad &p \geq \mu + \frac{1}{4\kappa} - a\\
p\left(\frac{1}{2} + 2 \kappa\left(\mu - p-a\right)\right), \quad &\mu - \frac{1}{4\kappa} - a \leq p \leq \mu + \frac{1}{4\kappa} - a\\
p, \quad &p \leq \mu - \frac{1}{4\kappa} - a\\
\end{cases} \text{ ,}\]
and so its optimal price is 
\[\label{monopsol}\tag{$4$}p_{M} = \begin{cases}
\mu - \frac{1}{4\kappa} - a, \quad &\mu \geq \frac{3}{4\kappa} + a\\
\frac{1}{8\kappa} + \frac{\mu - a}{2}, \quad &\frac{3}{4\kappa} + a \geq \mu \geq a - \frac{1}{4\kappa}
\end{cases} \text{ .}\]

The maximal profit, $\Pi_{M}^{*}$, and corresponding consumer payoff, $\Phi_{M}^{*}$, are
\[\tag{$5$}\label{exp8}\Pi_{M}^{*} = \begin{cases}
\mu - \frac{1}{4\kappa} - a, \quad &\mu \geq \frac{3}{4\kappa} + a\\
\frac{\left(1+4\kappa\left(\mu - a\right)\right)^2}{32\kappa}, \quad &\frac{3}{4\kappa} + a \geq \mu \geq a - \frac{1}{4\kappa}\\
0, \quad &\mu \leq a - \frac{1}{4\kappa}
\end{cases}, \text{ and } \quad \Phi_{M}^{*} = \begin{cases}
\frac{1}{4\kappa} + a, \quad &\mu \geq \frac{3}{4\kappa} + a\\
\frac{\Pi_{M}^{*}}{2}+a, &\frac{3}{4\kappa} + a \geq \mu\\
\end{cases} \text{ .}\]

These findings are quite intuitive: if the expected quality (net of the outside option), $\mu - a$, is sufficiently high; or, equivalently, information frictions are sufficiently high, the monopolist prices so as to dissuade learning. If $\mu-a$ is in a middle region, or $\kappa$ is sufficiently low, the monopolist wants to encourage learning. In fact, as $\kappa$ becomes increasingly small, the monopolist prices ever higher, to take advantage of the consumers who learn that they ``absolutely must have the product.'' 

\section{The Firms' Pricing Game}\label{infinite}

Having pinned down consumers' behavior, we may now describe the equilibrium pricing decisions of the firms. Following much of the literature, we assume that each firm's price is hidden and can only be discovered by a consumer when she visits the firm. To determine her optimal search and information acquisition protocol, the consumer conjectures a price, $\tilde{p}$ set by each firm, and at equilibrium these conjectures must be correct, i.e., $\tilde{p} = p$. 

For tractability, we focus on symmetric, pure strategy equilibria with random search.\footnote{Henceforth, we drop the modifier symmetric.} As is standard, we require that if the consumer observes a deviation at a firm, her beliefs about the other firms' prices are unchanged.\footnote{This is a standard assumption in the literature, though one that may not be realistic in situations with vertical relations (see, e.g., \cite*{janssen2020beliefs}). However, these are absent from our model, so we claim that our handling of off-path beliefs is appropriate.}

We begin by searching for a low price equilibrium, in which consumers search actively. A firm's profit is
\[\Pi\left(p\right) = \begin{cases}
0, \quad &\Tilde{p} + \sqrt{\frac{c}{\kappa}} \leq p\\
2\kappa p\left(\tilde{p}-p + \sqrt{\frac{c}{\kappa}}\right), \quad &\Tilde{p} + \frac{1}{2\kappa} - \sqrt{\frac{c}{\kappa}}  \leq p \leq \Tilde{p} + \sqrt{\frac{c}{\kappa}}\\
p, \quad &p \leq \Tilde{p} + \frac{1}{2\kappa} - \sqrt{\frac{c}{\kappa}}
\end{cases} \text{ .}\]
Taking the first order condition and imposing symmetry, we obtain the equilibrium price $p_{l} = \sqrt{c/\kappa}$ and substituting this into $\Phi$, we see that the consumer prefers to continue her search over her outside option ($0$) if and only $\mu \geq 2 \sqrt{c/\kappa} - 1/\left(4\kappa\right)$.
A consumer's payoff, a firm's profit, and a consumer's probability of stopping at the current firm are, respectively,
\[\Phi_{l}^{*} = \mu+\frac{1}{4\kappa}+c-2\sqrt{\frac{c}{\kappa}}, \quad \Pi_{l}^{*} = 2 c, \quad \text{and} \quad \mathbb{P}\left(u_{H}\right) = \frac{1}{2} + 2\kappa\left(\mu - p-a\right) = 2\sqrt{c\kappa} \text{ .}\]

Next, it is easy to see that the medium price equilibrium--in which there is no active search--is given by the interior monopoly solution when consumers have an outside option of $0$. Namely, $p_{m} = 1/\left(8\kappa\right) + \mu/2$, and the necessary and sufficient condition for this equilibrium to exist (that the consumer prefers to take her outside option over continuing her search and that trade occurs) is $2\sqrt{c/\kappa} - 1/\left(4\kappa\right) \geq \mu \geq -1/\left(4\kappa\right)$. A consumer's payoff and a firm's profit are, respectively,
\[\Phi_{m}^{*} = \frac{\left(1+4\kappa\mu\right)^{2}}{64\kappa}, \quad \text{and} \quad \Pi_{m}^{*} = \frac{\left(1+4\kappa\mu\right)^{2}}{32\kappa} \text{ .}\]

Compiling these observations,
\begin{proposition}
If $\mu \leq - 1/\left(4\kappa\right)$, no equilibria with trade exist. If $- 1/\left(4\kappa\right) \leq \mu \leq 2\sqrt{c/\kappa}- 1/\left(4\kappa\right)$, the unique equilibrium is one in which consumers visit only one firm and learn there (but do not search actively). The equilibrium price is $p_{m}$. If $2\sqrt{c/\kappa}- 1/\left(4\kappa\right) \leq \mu$, the unique equilibrium is one in which consumers both search and acquire information. The equilibrium price is $p_{l}$.
\end{proposition}
One noteworthy aspect of this result is that there exist no equilibria in which consumers neither search nor acquire information. This stands in sharp contrast to the monopolist scenario, in which the monopolist may wish to set a price so low that the consumer buys without learning. This cannot happen in the search market, since lower prices increase the consumer's outside option (her continuation value), thereby encouraging both learning and search.

\subsection{Comparative Statics}

At last, we can tackle some of the questions raised in the introduction. First, we observe the unsurprising result that an increase in the prior expected quality of the good, $\mu$, improves welfare in a Pareto sense. Naturally, prices are also increasing in $\mu$.

\begin{proposition}\label{mucompstat}
Increasing $\mu$ is Pareto improving, possibly strictly. If $\mu$ is sufficiently high (so there is active search), the market price and a firm's payoff are unaffected by $\mu$, but a consumer's payoff is strictly increasing in $\mu$. If $\mu$ is moderate (so there is no active search), prices, profits and consumer payoffs are all strictly increasing in $\mu$.
\end{proposition}

Second, we find that an increase in the explicit search cost, $c$, harms consumers but benefits firms.

\begin{proposition}\label{ccompstat}
If $c$ is sufficiently small (so there is active search), the market price and a firm's payoff are strictly increasing and a consumer's payoff is strictly decreasing in $c$. Otherwise (when there is no active search), the market price, a firm's payoff and a consumer's payoff are all unchanging in $c$.
\end{proposition}

It is easy to see that firms require explicit search costs in order to make profits. Indeed, observe that the lower posterior belief $u_{L}$ equals $\mu  - \sqrt{c/\kappa}$, which is decreasing in $c$ and goes to $\mu$ as $c$ vanishes. As costs vanish, the willingness of the consumer to depart for a different firm increases, eroding the rents that each firm can extract.

The effect of a change in $\kappa$ is more subtle. 

\begin{proposition}\label{compstatkappa}
Prices are strictly decreasing in $\kappa$. If $\kappa$ is sufficiently small, welfare is Pareto decreasing in $\kappa$. If $\mu$ is sufficiently large (or $c$ is sufficiently small) and $\kappa$ is sufficiently large, welfare is Pareto increasing in $\kappa$, possibly strictly (if the parameters are such that there is no active search). 
\end{proposition}

This result is somewhat surprising: it is possible that an increase in information frictions can be Pareto improving. The effect of these frictions on the price seems natural:  notice that when the information friction $\kappa$ is higher, $\mathbb{P}(u_H)$ is more responsive to changes in $p$, holding the continuation value fixed--a firm's ``demand curve'' is steeper. Consequently, at equilibrium, the solution of the optimization problem of a firm entails lowering the price slightly as a reaction to an increase in $\kappa$. Moreover, as it turns out, this balance is perfect: in regions of the parameter space with active search, a firm's profit is unaffected by changes in $\kappa$ since, as noted above, the rents it can extract are limited by the explicit search cost. 

In contrast, a change in the information acquisition cost affects a consumer's welfare in two ways. \textit{Ceteris Paribus}, an increase in $\kappa$ hurts consumers, as information is more difficult to obtain. However, it also results in a decrease in the market price, which is to a consumer's benefit. The overall effect of $\kappa$ on consumer welfare; therefore, is determined by which of these two forces dominates. When $\kappa$ is low (and therefore prices are high) the effect of an increase in $\kappa$ on the price is too small to outweigh the negative direct affect on learning. Conversely, a high $\kappa$ flips the relationship.

It is easy to see that if there is active search, the expected search duration is $1/\left(2\sqrt{c\kappa}\right)$, which is predictably decreasing in $c$ and $\kappa$. If there is no active search, the expected search duration is just $1$. 

\begin{figure}[tbp]
\centering
\begin{subfigure}{.5\textwidth}
  \centering
  \includegraphics[scale=.15]{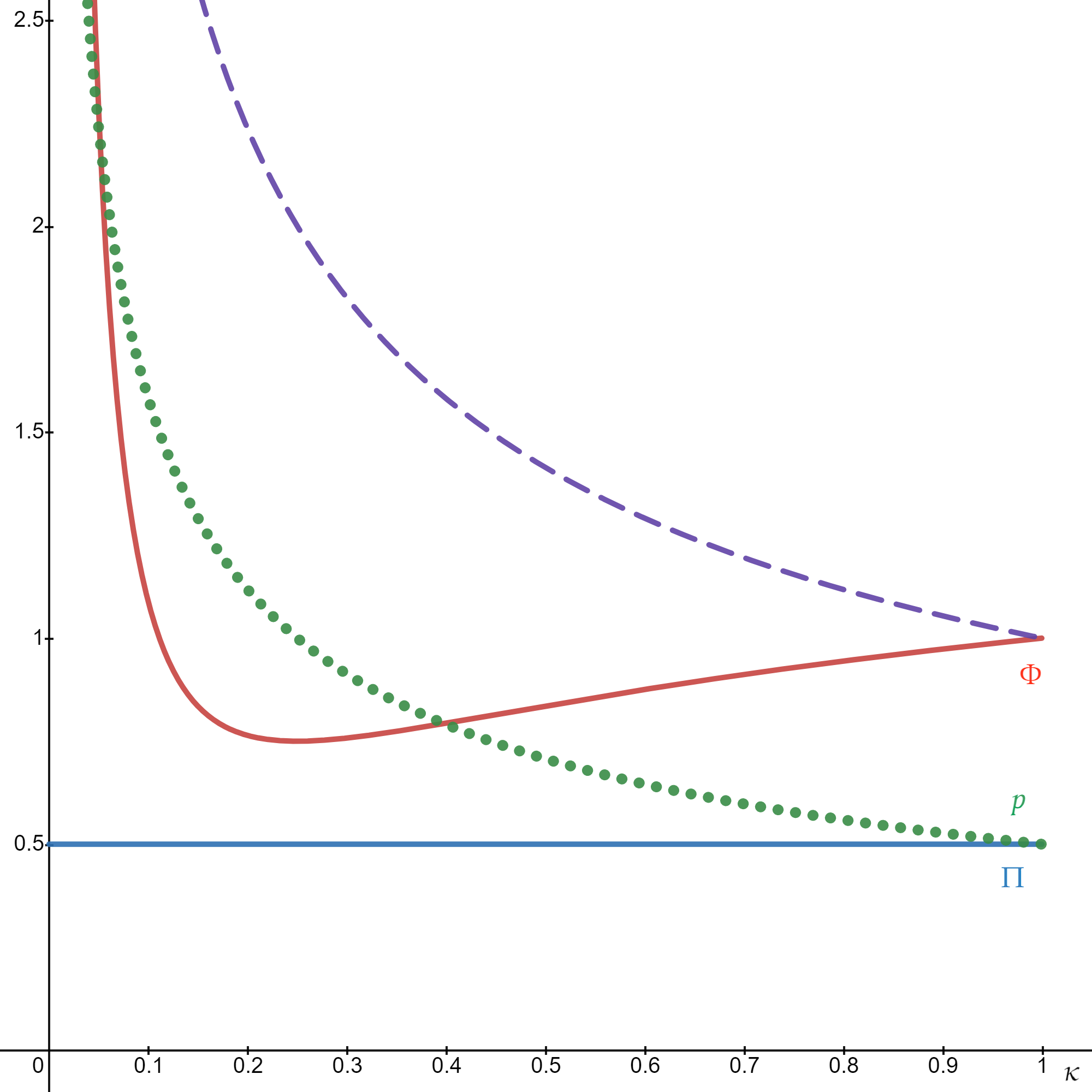}
  \caption{\textbf{Case 1:} $\mu > 0$ and $c \leq \mu/4$.}
  \label{figsub1}
\end{subfigure}%
\begin{subfigure}{.5\textwidth}
  \centering
  \includegraphics[scale=.15]{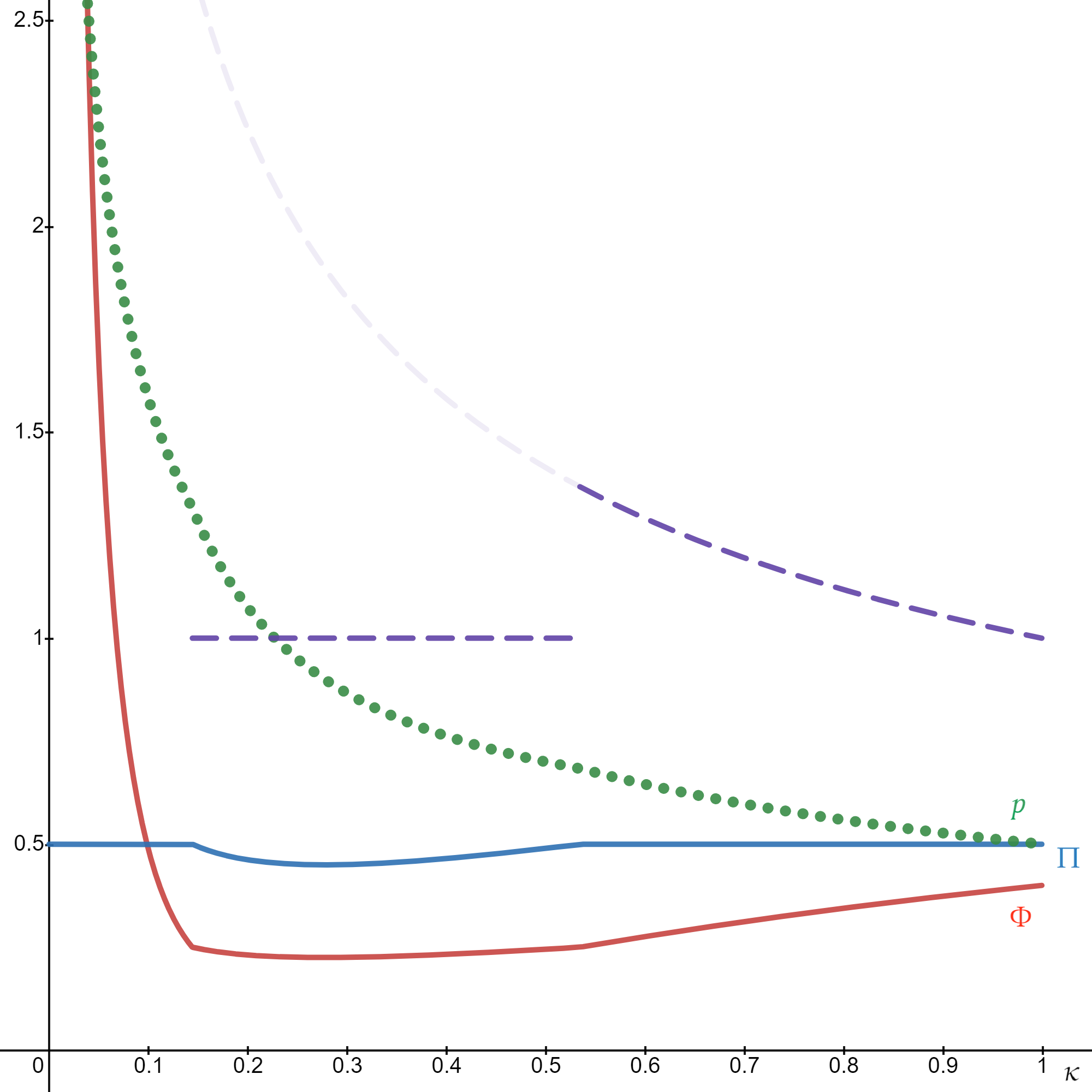}
  \caption{\textbf{Case 2:} $\mu > 0$ and $\mu/3 \geq c \geq \mu/4$.}
  \label{figsub2}
\end{subfigure}
\par
\bigskip
\par
\bigskip
\par
\begin{subfigure}{.5\textwidth}
  \centering
  \includegraphics[scale=.15]{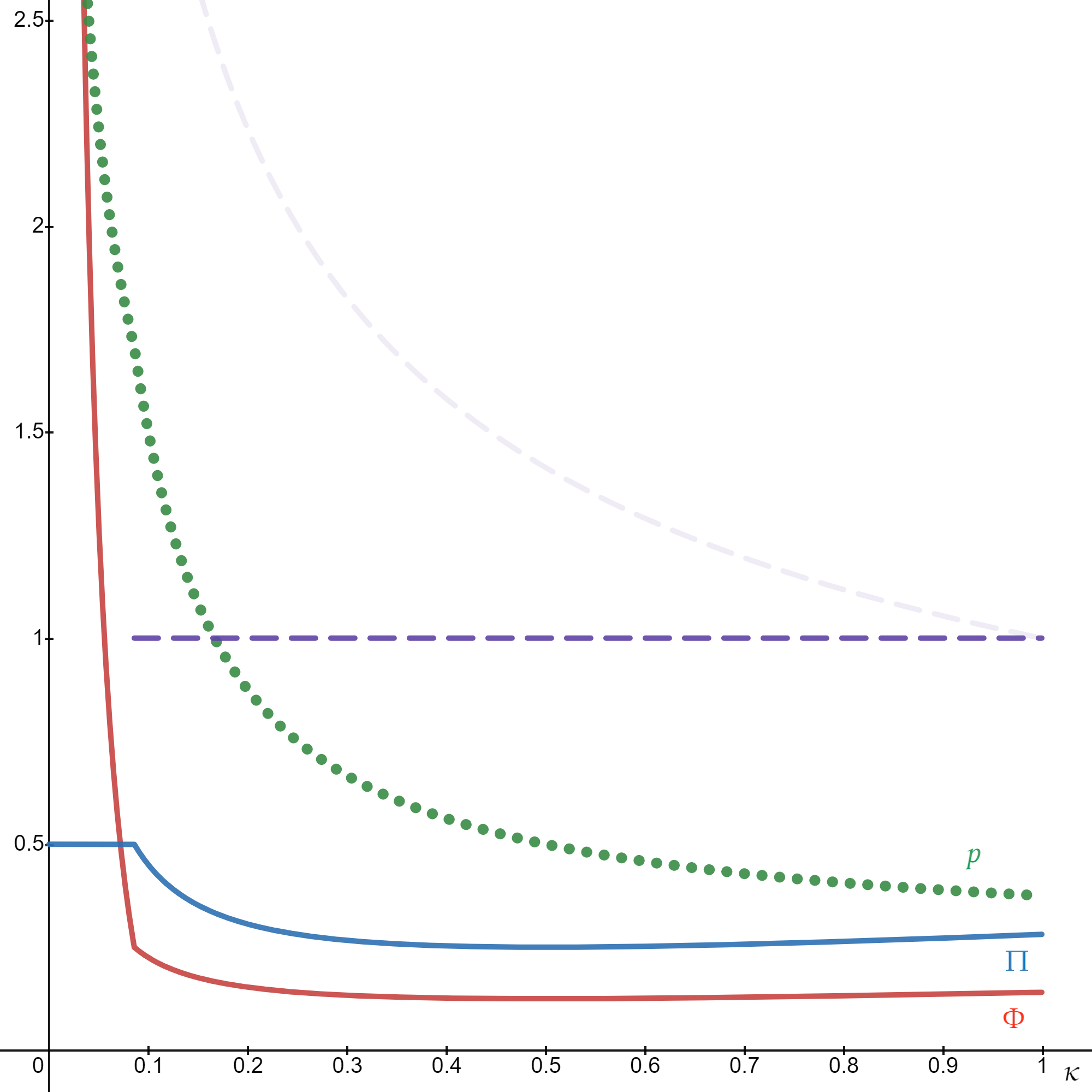}
  \caption{\textbf{Case 3:} $\mu > 0$ and $c \geq \mu/3$.}
  \label{figsub3}
\end{subfigure}
\caption{Consumer welfare (solid red), profit (solid blue), price (dotted green), and expected search duration (dashed purple) as functions of $\kappa$.}
\label{fig2}
\end{figure}

\begin{proposition}\label{searchdur}
The expected duration of search as a function of $\mu$ is increasing in $\mu$: it is a step function with one discontinuity at $\mu = 2\sqrt{c/\kappa} - 1/\left(4\kappa\right)$, where it jumps up from $1$ to $1/\left(2\sqrt{c\kappa}\right)$. The expected duration of search is decreasing in $c$: it is strictly decreasing in $c$ for $c$ sufficiently small (in the active search region) and constant thereafter (in the region without active search). The expected duration of search may be non-monotone in $\kappa$. In particular, for a region of the parameter space, there is active search for $\kappa$ sufficiently low or sufficiently high, but no active search in an intermediate region.
\end{proposition}

Figure \ref{fig2} illustrates the results of Propositions \ref{compstatkappa} and \ref{searchdur}. In the proof of the former (in Appendix \ref{compstatkappaproof}) we find that the parameter space (for $\mu > 0$) can be divided into three cases, each of which is depicted in the figure. The portions of the domain in which $\Pi$ is flat are those regions of $\kappa$ (for the fixed $\kappa$ and $\mu$) that facilitate active search. Note that the graph exhibits considerable non-monotonicity in $\kappa$: over different intervals of $\kappa$ profit and consumer welfare may both be increasing or decreasing. Moreover, the equilibrium may involve active search for small and large $\kappa$, but not for moderate $\kappa$. Accordingly, the expected search duration may have discrete jumps both up and down.

\subsection{Obfuscation}

By now, a sizeable collection of papers on obfuscation in search markets has emerged. These papers, both theoretical\footnote{See, e.g., \cite*{asher}, \cite*{gu2014strategic}, \cite*{hamalainen2018competitive}, and \cite*{petrikaite2018consumer}.} and empirical\footnote{See, e.g., \cite*{ellison2009search} and \cite*{richards2019strategic}. In addition, \cite{ellisonsurvey} discusses obfuscation in her survey of price search.} explore how firms may strategically increase consumers' search costs to allow for higher mark-ups. Instead of firms obfuscating by increasing search costs (the typical type of obfuscation explored in the literature), a different form of obfuscation suggests itself in our paper. Namely, we may allow each firm to not only set a price but also choose a level of information friction $\kappa$ from some interval of possible frictions $\left[\ubar{\kappa},\bar{\kappa}\right]$, where $\ubar{\kappa} > 0$ and $\bar{\kappa} \leq \infty$. 

How might firms be able to choose $\kappa$? Recall that in a consumer's dynamic problem, she observes in sequence each attribute, which is a symmetric (binary) random variable. By adding noise to this random variable, a firm would decrease the variance of the corresponding Brownian process, thereby increasing $\kappa$. Firms could also increase the explicit cost of inspecting each attribute, $\gamma$, which also raises $\kappa$.

A pair $(\kappa,p)$ is an equilibrium provided it is the solution to a firm's monopoly problem with $a$ determined by the equilibrium; and so looking at the monopolist's profit (given in Expression \ref{exp8}), we see that the equilibrium $\kappa$ must be a corner solution. \textit{Viz.},
\begin{remark}
If firms may choose both $p$ and $\kappa$, the equilibrium $\kappa^{*} \in \{\ubar{\kappa},\bar{\kappa}\}$.
\end{remark}

It is straightforward, though rather tedious to characterize the explicit equilibria of this game, but we believe the main insight we discover regarding obfuscation is contained in the remark and so do not do so. As the remark states, when firms may strategically obfuscate the product value, by hindering the observability of a product's attributes, they either muddle this process as much as possible or as minimally as possible. An interior level of obfuscation never constitutes an equilibrium. 
\section{Hidden Prices}\label{hidden}

So far, we have assumed that when a consumer visits a firm, she observes the price for the item \textit{before} acquiring information. This seems like reasonable assumption for many sorts of consumer goods like, e.g., shoes, bicycles, perfumes, headphones, foodstuffs. However, there are other markets in which it may be more natural to assume that a consumer acquires information about a product before seeing its price. One way this may manifest is when firms include hidden fees or charges that can only be observed later on in the purchase process.

In this section, we analyze the situation in which a consumer may not observe a firm's price until after she has acquired information at the firm.\footnote{Importantly, each firm sets its price in advance and may not observe consumers' behavior beforehand.} As we will shortly see, this changes the equilibrium behavior significantly: there is now an informational hold-up problem, which any equilibrium must accommodate.

As in our main specification, we begin by stating the outcome when there is a single (monopolist) firm. 

\subsection{Monopoly With Hidden Prices}

As we did in the observable prices case, we stipulate that a consumer has an outside option of $a \geq 0$. In contrast to the observable prices scenario, we find that the existence of a pure strategy pricing equilibrium for the monopolist is incompatible with trade.

\begin{lemma}\label{lemmanopure}
There are equilibria in which the monopolist chooses a deterministic price if and only if $\mu \leq a$. In any such pure-strategy equilibrium, there is no trade. 
\end{lemma}

If the monopolist mixes over prices, it must be indifferent over all prices in the support of her mixture. That is, given a consumer's distribution over posteriors, $F$, the monopolist's profit is given by
\[\hat{\Pi}\left(p\right) = p\left(1-F\left(p+a\right)\right) \text{ ,}\]
which equals some constant $\lambda > 0$ for all $p$ in the support of her mixed strategy. 

\begin{theorem}\label{monopvalue}
If $\mu > a$, there is an essentially unique\footnote{That is, any equilibrium must be such that the empirical distribution over valuations is $F$. One way this could be obtained is if each consumer randomized and chose distribution $F$ herself. Because there is a unit mass of consumers, one might worry about whether one could construct a continuum of independent individual random variables that yields the desired empirical distribution over values by the Exact Law of Large Numbers. \cite*{sun2006exact} allows for precisely such a construction to be done rigorously. Alternatively, we could assume that there is a random variable $Y \sim H$ with support $[0,1]$, which is a consumer's label. By, e.g., \cite*{wink} there is a distribution $H$ such that a consumer with label $y$ (deterministically) chooses a binary distribution over values with support $\left\{u_{L}\left(y\right), u_{H}\left(y\right)\right\}$ such that the resulting compound distribution is precisely $F$. Yet another approach would be to assume that there is a single representative consumer, who herself mixes over values according to $F$.} equilibrium in the monopolist scenario. The empirical distribution over consumer valuations is
\[F\left(x\right) \coloneqq 1 - \frac{\ubar{x}_{M}}{x}, \quad \text{on} \quad \left[\ubar{x}_{M},\bar{x}_{M}\right] \text{ ,}\]
and the monopolist chooses a uniform distribution over prices 
\[G\left(p\right) \coloneqq 2 \kappa\left(p - \ubar{p}_{M}\right), \quad \text{on} \quad \left[\ubar{p}_{M},\bar{p}_{M}\right] \text{ .}\]
$\ubar{p}_{M}$ solves
\[\label{eq8}\tag{$6$}\ln{\left\{1 + \frac{1}{2\kappa\ubar{p}_{M}}\right\}} = \frac{\mu-a}{\ubar{p}_{M}} - 1 \text{ ,}\]
and $\bar{p}_{M} = \ubar{p}_{M} + 1/\left(2\kappa\right)$, and $\ubar{x}_{M} = \ubar{p}_{M} + a$. The equilibrium vector of payoffs is 
\[\hat{\Pi}_{M}^{*} = \ubar{p}_{M}, \quad \text{and} \quad \hat{\Phi}_{M}^{*} = \kappa \left(\mu - \ubar{x}_{M}\right)^2 + a \text{ .}\]
\end{theorem}
Unsurprisingly, a truncated Pareto distribution ($F$) over valuations emerges, which yields unit elastic demand for the monopolist, which is, therefore, willing to mix over prices. The monopolist chooses a distribution over prices that yields a payoff for consumers (net of information acquisition costs) that is an affine function of the realized posterior values. They, in turn, are therefore willing to choose any distribution with support on $\left[\ubar{x}_{M},\bar{x}_{M}\right]$. Indeed, as the concavification approach suggests, if a consumer's payoff were strictly convex in her value, a binary distribution would be uniquely optimal. Conversely, if her payoff were strictly concave in her value, she would not acquire any information.

\subsubsection{Monopoly Comparative Statics}
Just as how the equilibrium price when prices are observable increases in $\mu$ and $\kappa$, so does the lower bound for its distribution over prices when prices are hidden.
\begin{lemma}\label{monopprice}
The lower bound for the monopoly price, $\ubar{p}_{M}$, is strictly increasing in $\kappa$ and $\mu$.
\end{lemma}
However, because the monopolist's profit is precisely this lower bound, $\ubar{p}_{M}$, information frictions are unambiguously good for the monopolist when prices are hidden. In the limit, as $\kappa$ explodes, $\ubar{p}_{M}$ and $\bar{p}_{M}$ converge to $\mu - a$ and so the monopolist leaves consumers with zero rents. In contrast to the observable prices case--when consumer welfare is initially decreasing in $\kappa$, then increasing, then decreasing again (after entering the region in which the monopolist does not induce learning)--when prices are hidden, consumer welfare is initially increasing in $\kappa$ then decreasing. In the limit, as $\kappa$ explodes, both profit and consumer welfare are the same regardless of when consumers may observe the price.
\begin{lemma}\label{monopwelfarehid}
The monopolist's profit is strictly increasing in $\kappa$ and $\mu$. Consumer welfare is strictly increasing in $\mu$ and is (strictly) increasing in $\kappa$ if and only if $\ubar{p}_{M} \kappa \lesssim .337$, i.e., if and only if $\kappa$ is sufficiently low. 
\end{lemma}

\subsection{The Monopolistic Competition Equilibrium}

If there exists an equilibrium with learning and active search, the equilibrium is precisely that given in the monopolist scenario, where we substitute a consumer's continuation value, $\hat{\Phi}-c$, in for $a$. Some simple algebra yields the following:

\begin{theorem}
If $\mu > 0$, there is an essentially unique equilibrium in the monopolistic competition scenario. The equilibrium involves active search if and only if $\mu - \sqrt{c/\kappa} > \ubar{p}$, where the price $\ubar{p}$ solves
\[\ln{\left\{1 + \frac{1}{2\kappa\ubar{p}}\right\}} = \frac{1}{\ubar{p}}\sqrt{\frac{c}{\kappa}} \text{ .}\]
The empirical distribution over consumer valuations is
\[F\left(x\right) \coloneqq 1 - \frac{\mu - \sqrt{\frac{c}{\kappa}}}{x}, \quad \text{on} \quad \left[\mu - \sqrt{\frac{c}{\kappa}},\mu - \sqrt{\frac{c}{\kappa}}+\frac{1}{2\kappa}\right] \text{ ,}\]
and the firms choose uniform distributions over prices 
\[G\left(p\right) \coloneqq 2 \kappa\left(p - \ubar{p}\right), \quad \text{on} \quad \left[\ubar{p},\bar{p}\right] \text{ .}\]
The equilibrium vector of payoffs is \[\hat{\Pi} = \ubar{p}, \quad \text{and} \quad \hat{\Phi} = \mu - \ubar{p} + c - \sqrt{\frac{c}{\kappa}} \text{ .}\]

If $\mu - \sqrt{c/\kappa} \leq \ubar{p}$, the unique equilibrium is given in Theorem \ref{monopvalue}, with $a = 0$.
\end{theorem}

Because $\ubar{p}$ is only given implicitly, the condition that demarcates the active search equilibrium is difficult to parse (since it itself is in terms of the equilibrium object $\ubar{p}$). As a result, it is helpful to supplement the theorem with the following proposition.

\begin{proposition}\label{auxhelp}
For any search cost $c$, there is a $\tilde{\mu}\left(c\right) > 0$ such that if $\mu < \tilde{\mu}\left(c\right)$ the equilibrium does not involve active search irrespective of $\kappa$. If $\mu > \tilde{\mu}\left(c\right)$, there exist $\ubar{\kappa}\left(c,\mu\right) > 0$ and $\bar{\kappa}\left(c,\mu\right) < 1/\left(4c\right)$ (with $\bar{\kappa} > \ubar{\kappa}$) such that the equilibrium involves active search if and only if $\kappa \in \left[\ubar{\kappa},\bar{\kappa}\right]$.
\end{proposition}

Thus, we see that either $\mu$ is so low that there is no active search, regardless of $\kappa$; or it is sufficiently high that there is active search if and only if $\kappa$ lies in an intermediate region.

There are a number of interesting implications of this equilibrium. Strikingly, despite the homogeneity of the market, and even though search is not directed, the unique equilibrium exhibits price dispersion. Furthermore, there may also be significant diversity of consumer behavior. Namely, the empirical distribution over consumer valuations may be driven by varying levels of learning by different consumers: some learn a lot at each firm while others learn very little. The unique equilibrium; therefore, justifies a broad spectrum of different behavior by consumers, even though they are identical. 

\subsubsection{Monopolistic Competition Comparative Statics}

\begin{proposition}\label{costathid1}
The lower bound for the price, $\ubar{p}$, is strictly increasing in $c$. In the active-search region (if it exists), $\ubar{p}$ is independent of $\mu$ and is (strictly) decreasing in $\kappa$ if and only if $h\left(c,\kappa\right) \coloneqq \ln{\left(\sqrt{c \kappa}\right)} + 2 - 2\sqrt{c\kappa} \leq (<) \ 0$. Outside of the active-search region, $\ubar{p}$ is strictly increasing in $\kappa$ and $\mu$. 
\end{proposition}

A simple calculation (or a glance at its graph) reveals that $h$ is strictly increasing in $\kappa$. Moreover, $\lim_{\kappa \searrow 0} h = -\infty$ and $h\left(1/\left(4\kappa\right)\right) = \ln\left(1/2\right) + 1 > 0$, so $h$ has a unique root at some $\tilde{\kappa} \in \left(0, 1/\left(4c\right)\right)$.

Because a firm's profit is merely the lower bound for the support of the price distribution, the following result is evident.

\begin{corollary}
A firm's profit is strictly increasing in $c$. In the active-search region, a firm's profit is independent of $\mu$ and is (strictly) decreasing in $\kappa$ if and only if $h\left(c,\kappa\right) \leq (<) \ 0$. Outside of the active-search region, profits are strictly increasing in $\kappa$ and $\mu$.
\end{corollary}

A direct calculation reveals that a consumer's welfare is strictly decreasing in $c$. On the other hand, note that 
\[\hat{\Phi}'\left(\kappa\right) = - \ubar{p}'\left(\kappa\right) + \frac{1}{2\kappa}\sqrt{\frac{c}{\kappa}} \text{ ,}\]
and so if the price is decreasing in $\kappa$, a consumer's welfare is increasing in $\kappa$. Naturally, even if the price is increasing in $\kappa$, a consumer's welfare may still be increasing in $\kappa$, provided the price increase is not too large. We have the following lemma.

\begin{lemma}\label{payoffkappa}
A consumer's payoff, $\hat{\Phi}$, is strictly decreasing in $c$. In the active-search region, $\hat{\Phi}$ is independent of $\mu$ and is (strictly) increasing in $\kappa$ if and only if $\kappa \ (<) \leq \kappa'$, where $\kappa' \in \left(\tilde{\kappa}, 1/\left(4c\right)\right)$. Outside of the active search region, a consumer's payoff is (strictly) increasing in $\kappa$ if and only if $\kappa$ is sufficiently low.
\end{lemma}

\section{Who (if Anyone) Benefits From Hidden Prices?}\label{fees}

One interpretation of the ``hidden prices'' version of the model is as a market in which firms can introduce hidden fees late in the purchase. This issue has come to the forefront of recent policy debates; and in the United States of America alone, a number of bills have been proposed, or rules imposed, recently to restrict drip pricing and other forms of hidden (additional) prices.\footnote{In their background section, \cite*{santana2020consumer} include a nice discussion of such recent rules and proposed legislation. Additional papers with similar analyses of the effect of hidden fees and prices on consumer shopping behavior include \cite*{blake2018price} and \cite*{bradley2020hidden}. The modal explanation for the adverse effects of hidden fees is behavioral: and in particular price (and tax) salience is often cited. \textit{Viz.}, consumers do not fully perceive the extra price. In contrast, hidden prices affect behavior in our model strategically--our consumers are fully rational and understand the ramifications of the hold-up problem in which they find themselves.} Our paper; therefore, can provide another perspective on this debate, since we may compare the various market outcomes (profits, consumer welfare, and prices) when prices are observable before learning versus when they are not.

\textit{A priori} one might suppose that hidden prices are to consumers' detriment, since it affects their learning adversely--a consumer understands that learning that she loves the product may allow the firm to charge a high price to exploit this enthusiasm. Furthermore, it also removes the ability of consumers to react to price changes and thereby discipline firms. As we shortly discover, this prediction is correct: consumers prefer observable fees, both when there is a search market as well as when there is a single monopolistic seller of the product. It is less clear which information regime firms prefer. On one hand, firms might have more impunity to set higher prices since consumers have less flexibility to react. However, consumers are strategic and so may not learn as much as the firms would like when prices are hidden. 


We begin our analysis with the monopoly case.

\subsection{Hidden Versus Observable Prices With Monopoly}

For simplicity, we assume that the consumer's outside option is $0$. Recall that with observable prices, the payoffs for the monopolist and a consumer are, respectively,
\[\Pi_{M}^{*} = \begin{cases}
\mu - \frac{1}{4\kappa}, \quad &\mu \geq \frac{3}{4\kappa}\\
\frac{\left(1+4\kappa\mu\right)^2}{32\kappa}, \quad &\frac{3}{4\kappa} \geq \mu \geq - \frac{1}{4\kappa}\\
0, \quad &\mu \leq - \frac{1}{4\kappa}
\end{cases}, \quad \text{and} \quad \Phi_{M}^{*} = \begin{cases}
\frac{1}{4\kappa}, \quad &\mu \geq \frac{3}{4\kappa}\\
\frac{\Pi_{M}^{*}}{2}, &\frac{3}{4\kappa} \geq \mu\\
\end{cases} \text{ .}\]
In contrast, with hidden prices, the payoffs are
\[\hat{\Pi}_{M} = \begin{cases}
\ubar{p}_{M}, \quad &\mu > 0\\
0, \quad &\mu \leq 0
\end{cases}, \quad  \text{and} \quad \hat{\Phi}_{M} = \begin{cases}
\kappa \left(\mu - \ubar{p}_{M}\right)^2, \quad &\mu > 0\\
0, \quad &\mu \leq 0
\end{cases} \text{ ,}\]
where $\ubar{p}_{M}$ is given implicitly by
\[\ln{\left\{1 + \frac{1}{2\kappa\ubar{p}_{M}}\right\}} = \frac{\mu}{\ubar{p}_{M}} - 1 \text{ .}\]

Obviously, if the prior is too low ($\mu \leq -1/\left(4\kappa\right)$) there is no trade in either case. On the other hand, if $-1/\left(4\kappa\right) < \mu \leq 0$, there is trade if and only if prices are observed before learning. Thus, in that region of the parameter space, both monopolist and consumers are strictly better off when prices are not hidden.

If the prior is high ($\mu \geq 3/\left(4\kappa\right)$), the consumer is better off with observable prices if and only if 
\[\frac{1}{4\kappa} \geq \kappa \left(\mu - \ubar{p}_{M}\right)^2 \quad \Leftrightarrow \quad \ubar{p}_{M} \geq \mu - \frac{1}{2\kappa} \text{ .}\]
As we show in Appendix \ref{monopcompareproof}, this always holds. The monopolist is better off with observable prices if and only if $\ubar{p}_{M} \leq \mu - 1/\left(4\kappa\right)$. This must hold in this region of the parameter space as well.

For intermediate priors ($3/\left(4\kappa\right) \geq \mu > 0$), a consumer is better off with observable prices if and only if
\[\frac{\left(1+4\kappa\mu\right)^2}{64\kappa} \geq \kappa \left(\mu - \ubar{p}_{M}\right)^2 \quad \Leftrightarrow \quad \ubar{p}_{M} \geq \frac{\mu}{2} - \frac{1}{8\kappa} \text{ ,}\]
which is always true. Finally, the monopolist is better off with observable prices if and only if
\[\frac{\left(1+4\kappa\mu\right)^2}{32\kappa} \geq \ubar{p}_{M} \text{ ,}\]
which is also always true. 

The following proposition sums up our stark conclusion:
\begin{proposition}\label{monopcompare}
Both monopolist and consumers prefer prices to be observable before learning. 
\end{proposition}
Figure \ref{fig3} depicts the monopolist's profits with observable prices and hidden prices as well as its price when prices are observable and the lower bound of its distribution over prices when prices are hidden. In some sense, this result is unsurprising. By allowing the monopolist the power to commit to its price, consumers have more incentive to learn, since they know that they will not be exploited.

\begin{figure}
    \centering
    \includegraphics[scale=.33]{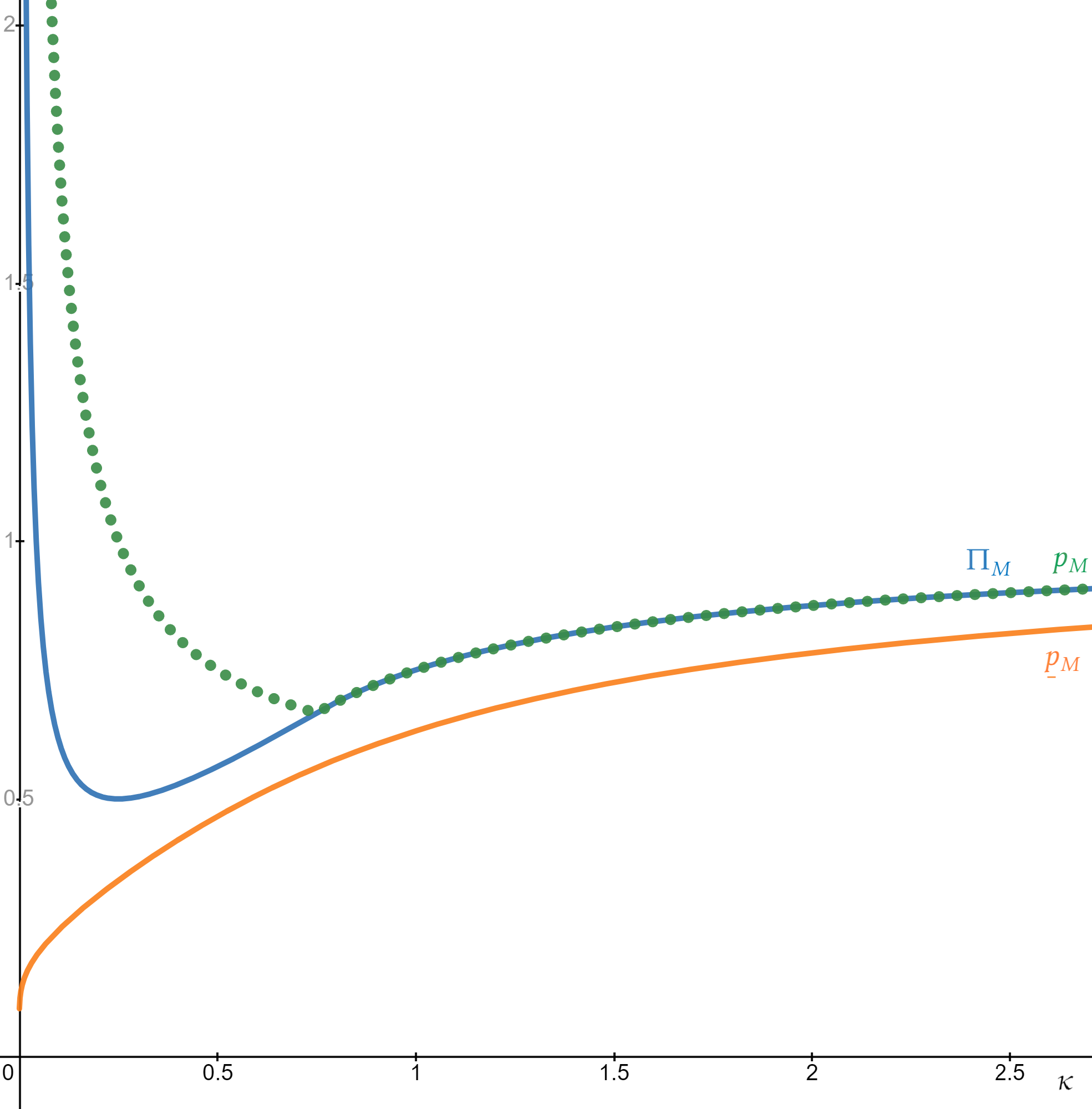}
    \caption{Monopolist profit with observable prices (solid blue), monopolist price with observable prices (dotted green), and monopolist profit and price lower bound with hidden prices (solid orange) as functions of $\kappa$.}
    \label{fig3}
\end{figure}

\subsection{Hidden Versus Observable Prices With Monopolistic Competition}

As we stated above, consumers prefer to observe prices before learning in the search market with multiple firms. The intuition is analogous to that when there is a single seller: they no longer have to worry about being exploited and can therefore learn optimally as a result.

\begin{proposition}\label{oligopco}
Consumer welfare when prices are observed before learning is strictly higher than when prices are observed after learning.
\end{proposition}

Surprising, the result from the monopoly scenario does not always carry over for firms. That is, it is possible that firms may actually prefer hidden prices. This occurs when $\mu$ and $\kappa$ are sufficiently high (in comparison to $c$). The explanation for this finding is as follows: if the parameters are such that there is active search when prices are hidden, there must be active search when prices are observed. Moreover, recall that when there is active search, a firm's profit is independent of both the prior quality and the level of information frictions and is determined only by $c$, which is quite low by assumption. In contrast, when prices are hidden (and there is active search), $c$ must be so low (or $\mu$ and $\kappa$ so high) that the lower bound of the equilibrium price, $\ubar{p}$ is greater than $2 c$, the profit in the observable prices set-up. Of course, $\ubar{p}$ is precisely a firm's profit when prices are hidden, and so its profits are, therefore, higher.

When the parameters are such that consumers do not actively search in either observation regime, we are back in the monopolist setting and so observable prices are preferred by the firms. The third and final possibility is that the equilibrium involves active search if and only if prices are observable. This is a mix of the other two cases: if $\mu$ and $\kappa$ are sufficiently high, firms prefer hidden prices, but not if they are too low.

\begin{figure}[tbp]
\centering
\begin{subfigure}{.5\textwidth}
  \centering
  \includegraphics[scale=.15]{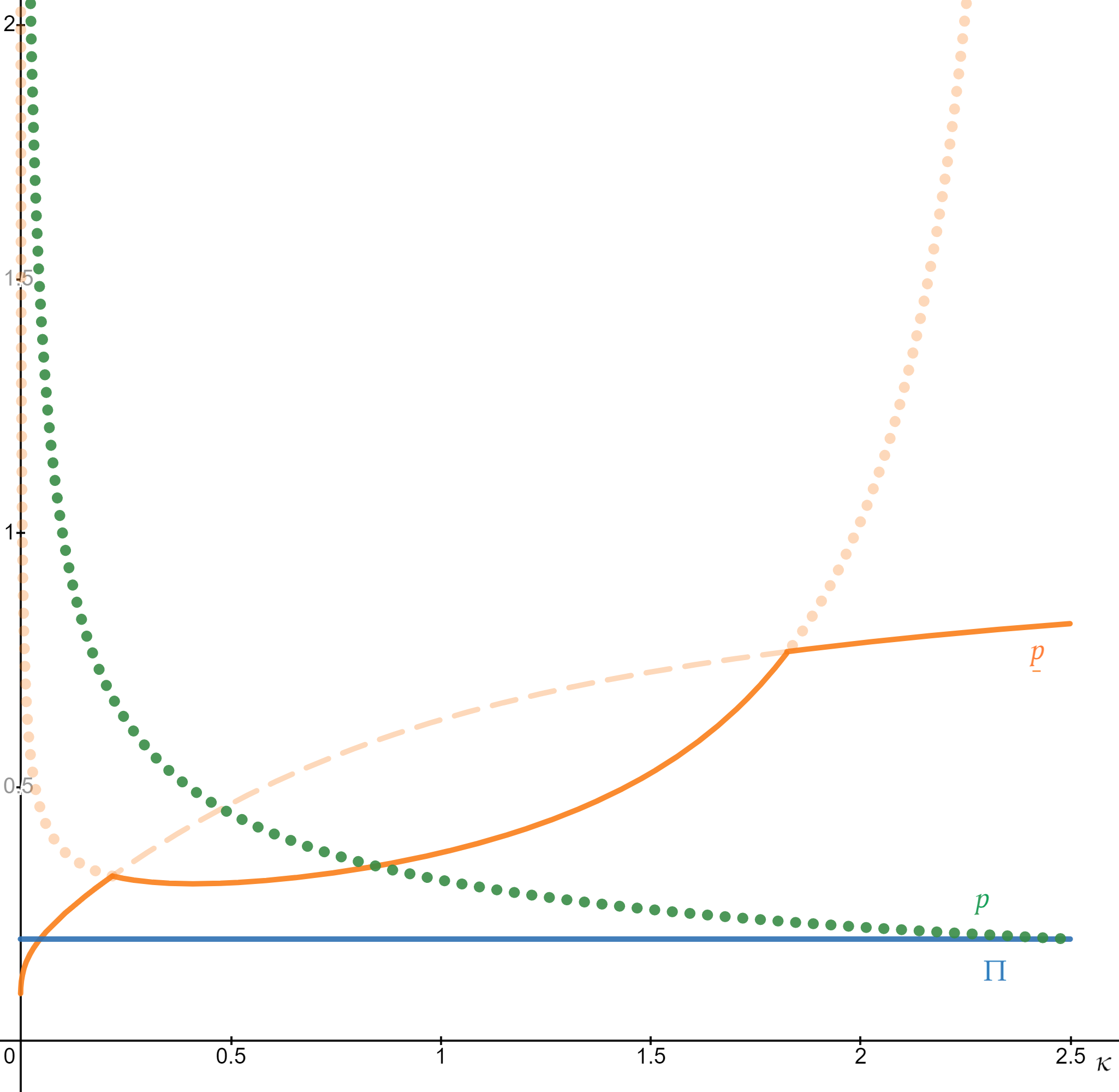}
  \caption{\textbf{Case 1:} $\mu > 0$ and $c \leq \mu/4$.}
  \label{figsub12}
\end{subfigure}%
\begin{subfigure}{.5\textwidth}
  \centering
  \includegraphics[scale=.15]{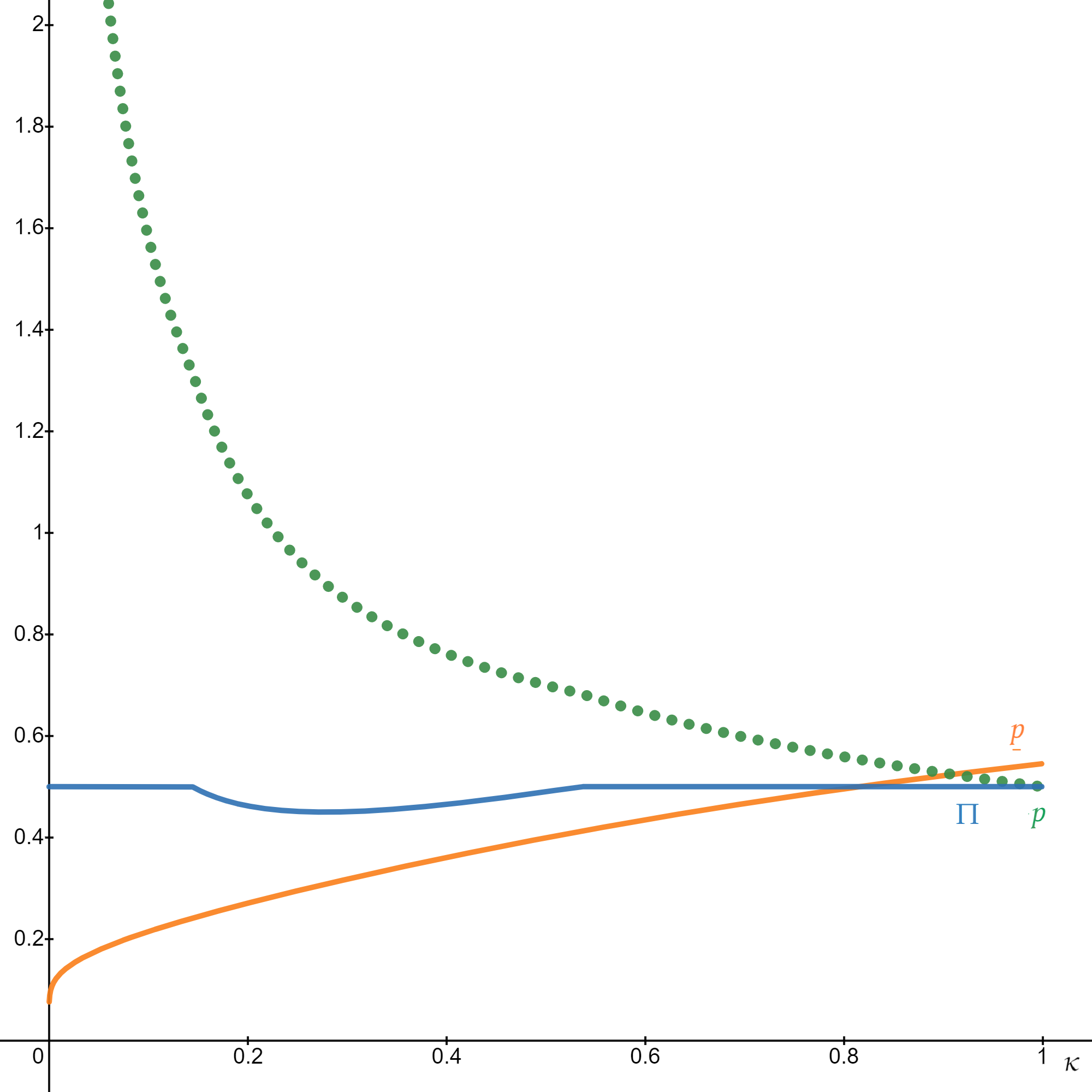}
  \caption{\textbf{Case 2:} $\mu > 0$ and $\mu/3 \geq c \geq \mu/4$.}
  \label{figsub22}
\end{subfigure}
\par
\bigskip
\par
\bigskip
\par
\begin{subfigure}{.5\textwidth}
  \centering
  \includegraphics[scale=.15]{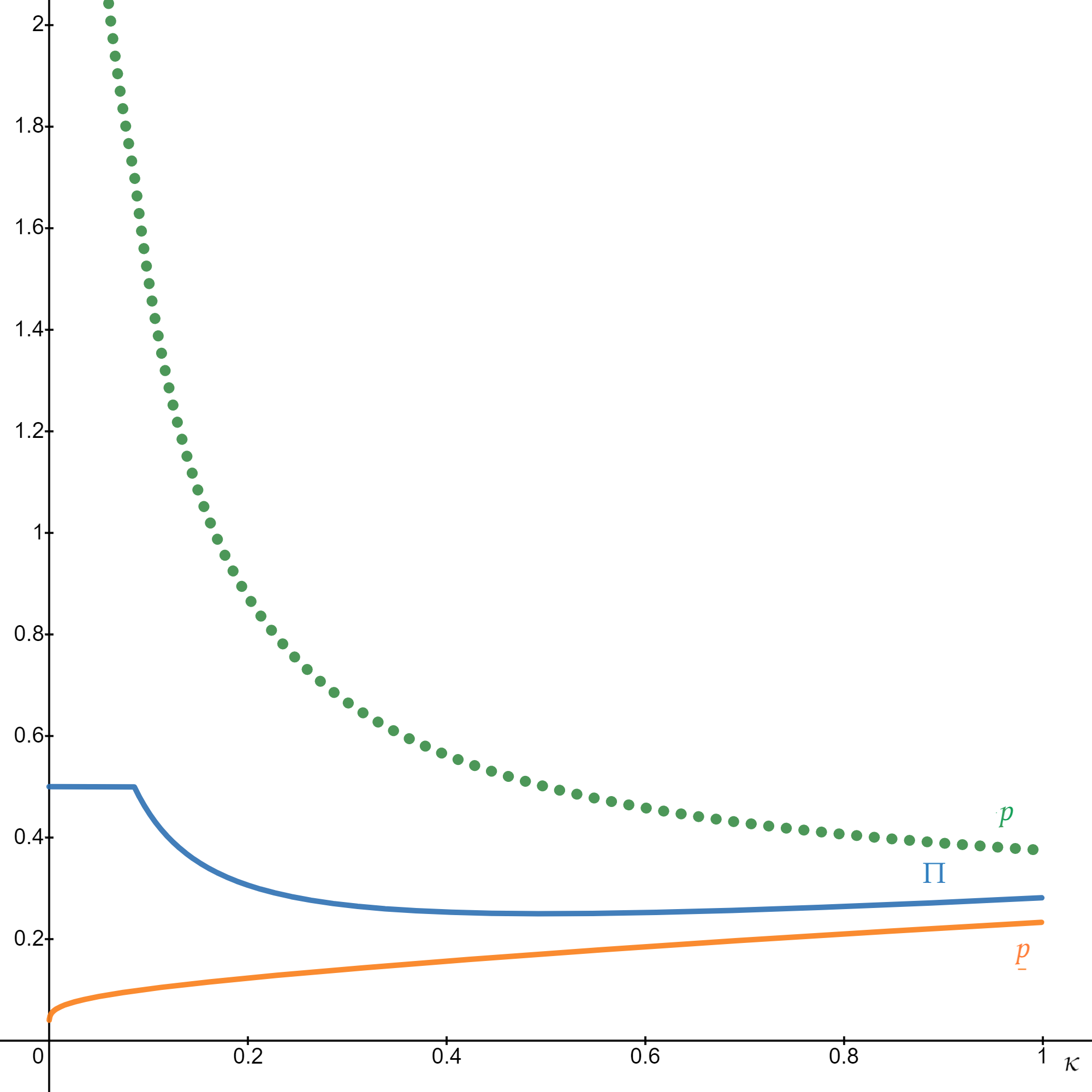}
  \caption{\textbf{Case 3:} $\mu > 0$ and $c \geq \mu/3$.}
  \label{figsub32}
\end{subfigure}
\caption{Profit with observable prices (solid blue), price with observable prices (dotted green), and profit and price lower bound with hidden prices (solid orange) as functions of $\kappa$. }
\label{fig4}
\end{figure}

\begin{proposition}\label{olicopco2}
Firms prefer that consumers observe prices before learning if and only if $\mu$ is sufficiently low.
\end{proposition}

Figure \ref{fig4} depicts firms' profits with observable prices and hidden prices as well as their prices when prices are observable and the lower bound of the market distribution over prices when prices are hidden. 
 
\section{Discussion}\label{discus}

With the exception of \cite{guo2021endogenous}, the existing literature that explores sequential consumer search does not allow for flexible information acquisition by consumers. Similarly, the existing rational inattention literature does not tackle the dynamic problem of consumers searching and flexibly acquiring information in sequence. We formulate and solve a model that incorporates these aspects, which allows us to identify and highlight the vital differences between search and information frictions. 

Despite the ability of consumers to adapt their learning to the prices they observe, higher search frictions benefit firms and lead to greater market prices. Increased search costs make consumers more reluctant to leave firms, making them \textit{less} responsive to price increases, which firms can exploit. In contrast, higher information frictions make consumers \textit{more} responsive to price increases, thereby decreasing market prices and potentially benefiting consumers. We find that these relationships hold (for the most part) even when consumers cannot observe prices before acquiring information: firms benefit at the expense of consumers as search costs increase, whereas an increase in information frictions may improve the lot of consumers.

Our analysis; therefore, suggests that mandated decreases in information frictions (mandatory disclosure rules or transparency requirements, say) may be to the market's detriment. Moreover, our comparison of cases in which a firm's price is observed before or after information acquisition allows us to provide input on policies banning hidden fees. We find that consumers are always better off when they can observe prices before learning but that firms may be hurt by the transparency, in spite of the hold-up problem that ensues when prices are hidden.


\bibliography{sample.bib}

\appendix
\section{Omitted Proofs and Derivations} \label{appendix1}

\subsection{Proposition \ref{mucompstat} Proof}
\begin{proof}
By direct computation, if $2\sqrt{c/\kappa}- 1/\left(4\kappa\right) > \mu \geq - 1/\left(4\kappa\right)$, the market price, a firm's payoff and a consumer's payoff are strictly increasing in $\mu$. Analogously, if $\mu \geq 2\sqrt{c/\kappa}- 1/\left(4\kappa\right)$, the market price and a firm's payoff are unaffected by $\mu$, but a consumer's payoff is strictly increasing in $\mu$.
\end{proof}

\subsection{Proposition \ref{ccompstat} Proof}
\begin{proof}
Recall that there is active search and information acquisition in equilibrium provided \[\label{activesearchin}\tag{$A1$} \mu \geq 2\sqrt{\frac{c}{\kappa}} - \frac{1}{4\kappa} \text{ .}\]
Rearranging, we get a number of cases.
\begin{flushleft}
\textbf{Case 1:} If $\mu \geq 3/\left(4\kappa\right)$, Inequality \ref{activesearchin} holds (since $0 < c < 1/\left(4\kappa\right)$).
\end{flushleft}

\begin{flushleft}
\textbf{Case 2:} If $ -1/\left(4\kappa\right) \leq \mu < 3/\left(4\kappa\right)$, Inequality \ref{activesearchin} holds if and only if
\[0 < c < \frac{\left(1+4\kappa \mu\right)^2}{64\kappa} \text{ .}\]
\end{flushleft}

\begin{flushleft}
\textbf{Case 3:} If $\mu \leq -1/\left(4\kappa\right)$, Inequality \ref{activesearchin} does not hold.
\end{flushleft}
In the active-search region, a simple calculation reveals that $\Phi^{*}_{l}$ is strictly decreasing in $c$. In this region, $\Pi^{*}_{l}$ and $p_{l}$ are strictly increasing in $c$. Trivially, these equilibrium objects are independent of $c$ in the region without active search.
\end{proof}

\subsection{Proposition \ref{compstatkappa} Proof}\label{compstatkappaproof}

\begin{proof}
In the active-search region, 
\[\frac{\partial{\Phi_{l}^{*}}}{\partial{\kappa}} = - \frac{1}{4\kappa^2} + \frac{1}{\kappa} \sqrt{\frac{c}{\kappa}} \geq \ (>) \ 0 \quad \Leftrightarrow \quad \kappa \geq \ (>) \ \frac{1}{16c} \text{ .}\]
In this region, $\Pi_{l}^{*}$ is unchanging in $\kappa$, and $p_{l}$ is strictly decreasing in $\kappa$. In the region without active search,
\[\frac{\partial{\Phi_{m}^{*}}}{\partial{\kappa}} = \frac{\mu^2}{4}-\frac{1}{64\kappa^2} \geq \ (>) \ 0 \quad \Leftrightarrow \quad \kappa \geq \ (>) \ \left|\frac{1}{4\mu}\right| \text{ .}\]
Since $\Pi_{m}^{*} = 2 \Phi_{m}^{*}$, the effect of a change of $\kappa$ on a firm's profit is identical to its effect on consumer welfare in this (no-search) region. $p_{m}$ is strictly decreasing in $\kappa$.

Rearranging Equation \ref{activesearchin}, we observe that there are four cases to exhaust:
\begin{flushleft}
\textbf{Case 1:} $\mu > 0$ and $0 < c \leq \mu/4$. Here, $\kappa < 1/\left(4 c\right)$ implies inequality \ref{activesearchin}. Consequently, 
\[\sgn{\Phi^{*'}\left(\kappa\right)} = \begin{cases}
-1, \quad &0 < \kappa < \frac{1}{16c}\\
0, \quad &\kappa = \frac{1}{16c}\\
1, \quad &\frac{1}{16c} < \kappa < \frac{1}{4c}
\end{cases}, \quad \text{and} \quad \sgn{\Pi^{*'}\left(\kappa\right)} = 0 \text{ .}\]
\end{flushleft}

\begin{flushleft}
\textbf{Case 2:} $\mu > 0$ and $\mu/4 < c < \mu/3$. Here, \ref{activesearchin} holds either if 
\[0 < \kappa \leq \frac{8 c - \mu - 4 \sqrt{\left(4 c - \mu\right)c}}{4 \mu^2} \eqqcolon \ubar{\kappa}, \quad \text{or} \quad \bar{\kappa} \coloneqq \frac{8 c - \mu + 4 \sqrt{\left(4 c - \mu\right)c}}{4 \mu^2} \leq \kappa < \frac{1}{4c} \text{ .}\]
Consequently, 
\[\sgn{\Phi^{*'}\left(\kappa\right)} = \begin{cases}
-1, \quad &0 < \kappa < \ubar{\kappa}\\
-1, \quad &\ubar{\kappa} \leq \kappa < \frac{1}{4 \mu}\\
0, \quad &\kappa = \frac{1}{4\mu}\\
1, \quad &\frac{1}{4\mu} < \kappa < \bar{\kappa}\\
1, \quad &\bar{\kappa} \leq \kappa < \frac{1}{4c}
\end{cases}, \quad \text{and} \quad \sgn{\Pi^{*'}\left(\kappa\right)} = \begin{cases}
0, \quad &0 < \kappa < \ubar{\kappa}\\
-1, \quad &\ubar{\kappa} \leq \kappa < \frac{1}{4 \mu}\\
0, \quad &\kappa = \frac{1}{4\mu}\\
1, \quad &\frac{1}{4\mu} < \kappa < \bar{\kappa}\\
0, \quad &\bar{\kappa} \leq \kappa < \frac{1}{4c}
\end{cases} \text{ .}\]
\end{flushleft}

\begin{flushleft}
\textbf{Case 3:} $\mu \geq 0$ and $c \geq \mu/3$. Here, \ref{activesearchin} holds if $0 < \kappa \leq \ubar{\kappa}$. Accordingly,
\[\sgn{\Phi^{*'}\left(\kappa\right)} = \begin{cases}
-1, \quad &0 < \kappa < \ubar{\kappa}\\
-1, \quad &\ubar{\kappa} \leq \kappa < \frac{1}{4 \mu}\\
0, \quad &\kappa = \frac{1}{4\mu}\\
1, \quad &\frac{1}{4\mu} < \kappa < \frac{1}{4c}
\end{cases}, \quad \text{and} \quad \sgn{\Pi^{*'}\left(\kappa\right)} = \begin{cases}
0, \quad &0 < \kappa < \ubar{\kappa}\\
-1, \quad &\ubar{\kappa} \leq \kappa < \frac{1}{4 \mu}\\
0, \quad &\kappa = \frac{1}{4\mu}\\
1, \quad &\frac{1}{4\mu} < \kappa < \frac{1}{4c}
\end{cases} \text{ .}\]
\end{flushleft}

\begin{flushleft}
\textbf{Case 4:} $\mu < 0$. Again, \ref{activesearchin} holds if $0 < \kappa \leq \ubar{\kappa}$. Accordingly,
\[\sgn{\Phi^{*'}\left(\kappa\right)} = \begin{cases}
-1, \quad &0 < \kappa < \ubar{\kappa}\\
1, \quad &\ubar{\kappa} \leq \kappa < \frac{1}{4c}
\end{cases}, \quad \text{and} \quad \sgn{\Pi^{*'}\left(\kappa\right)} = \begin{cases}
0, \quad &0 < \kappa < \ubar{\kappa}\\
1, \quad &\ubar{\kappa} \leq \kappa < \frac{1}{4c}
\end{cases} \text{ .}\]
\end{flushleft}

\end{proof}

\subsection{Lemma \ref{lemmanopure} Proof}

\begin{proof}
Let us begin by establishing that if $\mu > a$ there exist no pure-strategy equilibria. For convenience, we assume that there is a single representative consumer. Suppose for the sake of contradiction that there is an equilibrium in which the seller sets a deterministic $p$. Recall that if $\mu \in \left(a+p-1/\left(4\kappa\right), a+p+1/\left(4\kappa\right)\right)$, the consumer's learning has support $\left\{a+p-1/\left(4\kappa\right), a+p+1/\left(4\kappa\right)\right\}$. Otherwise, the consumer does not learn and either buys with certainty (if $\mu$ is sufficiently high) or takes her outside option with certainty. That leaves us with three cases.

Suppose first that $p \geq \mu - a + 1/\left(4\kappa\right)$, so that the consumer does not learn and takes her outside option (on path). The seller's profit is $0$. Evidently, it can deviate by charging some price $\Tilde{p} \in (0, \mu - a)$, which yields it profit $\Tilde{p} > 0$. 

Second, suppose that $p \in \left(\mu - a - 1/\left(4\kappa\right), \mu - a + 1/\left(4\kappa\right)\right)$. Conditional on the consumer purchasing, the seller's profit is $p$. Suppose it deviates to a price $\Tilde{p} \coloneqq p + 1/\left(8\kappa\right)$. Evidently, 
\[u_{H} - \tilde{p} = a + p + \frac{1}{4\kappa} - \left(p + \frac{1}{8\kappa}\right) = a + \frac{1}{8\kappa} > a \text{ ,}\]
so the consumer still purchases with the same probability at a strictly higher price, yielding a strictly higher profit for the seller.

Third, suppose that $p \leq \mu - a - 1/\left(4\kappa\right)$, so that the consumer does not learn and purchases from the seller. Evidently, the firm's profit is $p$. However, it can deviate to some $\tilde{p} = \mu - a - \epsilon$, which is strictly higher than $p$ and $0$ provided $\epsilon < \min\left\{1/\left(4\kappa\right), \mu - a\right\}$.

We have established the necessity of $\mu \leq a$ for the existence of a pure-strategy equilibrium. Now let us deduce the no-trade result.

Let $\mu \leq a$, and suppose for the sake of contradiction that there exists a pure-strategy equilibrium in which trade occurs. Consequently, either $p \in \left(\mu - a - 1/\left(4\kappa\right), \mu - a + 1/\left(4\kappa\right)\right)$ or $p \leq \mu - a - 1/\left(4\kappa\right)$. The latter is clearly impossible, since the seller is making negative profits. In the first case, we must have $p \leq \mu - a + 1/\left(4\kappa\right)$. Clearly, if $\mu - a + 1/\left(4\kappa\right) \leq 0$, there are no pure-strategy equilibria with trade. Thus, let us assume $\mu - a + 1/\left(4\kappa\right) > 0$. If there is trade (with positive probability), then $0 \leq p < \mu - a + 1/\left(4\kappa\right)$. But then the firm can raise its price slightly, thereby increasing its profit. 

Finally, let us construct pure strategy equilibria in which there is no trade. Evidently, we must have $p \geq \mu - a + 1/\left(4\kappa\right)$. The consumer will not learn and so (on-path) the consumer will take her outside option. On path, the firm gets $0$ and any deviation would yield it at most $0$ since $\mu \leq a$.
\end{proof}

\subsection{Theorem \ref{monopvalue} Proof}

\begin{proof}
Again, for convenience, we assume that there is a single representative consumer. From Lemma \ref{lemmanopure}, we know that there exist no pure strategy equilibria in this region. As noted in the main part of the text, given the consumer's equilibrium distribution over values, $F$, the monopolist's profit is \[\hat{\Pi}\left(p\right) = p\left(1-F\left(p+a\right)\right) \text{ ,}\]
which must equal some constant $\lambda \geq 0$ for all $p$ in the support of its mixed strategy since it is mixing. Denote the support of its mixed strategy by $\left[\ubar{p},\bar{p}\right]$, where $\ubar{p} \geq 0$. Given the monopolist's mixed strategy over prices $G$, the consumer's distribution ($F$) must solve
\[\max_{F}\int_{-\infty}^{\infty}\left\{\int_{\ubar{p}}^{x-a}\left(x-p\right)dG\left(p\right) + \int_{x-a}^{\bar{p}}adG\left(p\right) - \kappa \left(x-\mu\right)^2\right\}dF\left(x\right) \text{ .}\]
It is also helpful to write down the consumer's payoff as a function of the realized value:
\[V\left(x\right) = 
\begin{cases}
a - \kappa \left(x-\mu\right)^2, \quad &x < a + \ubar{p}\\
a \left(1-G\left(x-a\right)\right) + \int_{\ubar{p}}^{x-a}\left(x-p\right)dG\left(p\right) - \kappa \left(x-\mu\right)^2, \quad &a + \ubar{p} \leq x \leq a + \bar{p}\\
x - \mathbb{E}_{G}\left[p\right] - \kappa \left(x-\mu\right)^2, \quad &a + \bar{p} \leq x\\
\end{cases} \text{ .}\]
Concavifying $V$ yields the consumer's optimal distribution(s). Our first pair of claims ensures that the upper and lower bounds of the support of any equilibrium distributions over values and prices differ from each other (respectively) only by the value of the outside option.
\begin{claim}
In any equilibrium, the upper bound of the support of the consumer's distribution, $\bar{x}$, satisfies $\bar{x} = a + \bar{p}$.
\end{claim}
\begin{proof}
It is easy to see that we cannot have $\bar{x} < a + \bar{p}$. Indeed, the monopolist's profit from any price $p > \bar{x} - a$ is $0$, whereas it can always obtain a strictly positive profit by setting price $p' \in \left(0,\mu-a\right)$. Can we have $\bar{x} > \bar{p} + a$? In this case the consumer's optimal distribution above $\bar{p} + a$ must consist of a single mass point, since information is costly. However, then we have $F\left(\bar{p} + a\right) = F\left(\bar{p} + a + \eta\right)$ for all (sufficiently small) strictly positive $\eta$, contradicting the optimality of the monopolist's distribution over prices.
\end{proof}
\begin{claim}
In any equilibrium, the lower bound of the support of the consumer's distribution, $\ubar{x}$, satisfies $\ubar{x} = a + \ubar{p}$.
\end{claim}
\begin{proof}
That $\ubar{x} \leq a + \ubar{p}$ is also straightforward. Otherwise, the monopolist obtains a strictly higher profit from $\ubar{p} + \epsilon$ than from $\ubar{p}$, a violation. Can we have $\ubar{x} < \ubar{p} + a$? In this case the consumer's optimal distribution below $\bar{p} + a$ must consist of a single mass point, since information is costly. Moreover, $F$ must have no support on $\left(\ubar{x}, \tilde{x}\right)$, where $\tilde{x} > \ubar{p} + a$ since $V$ is strictly concave on $\left(-\infty, \ubar{p} + a\right)$. However, this means that the monopolist strictly prefers price $\ubar{p} + \eta$ to $\ubar{p}$ since demand is locally perfectly inelastic in a neighborhood of $\ubar{p}$.
\end{proof}
We now establish that the consumer's distribution over values has full support.
\begin{claim}
In any equilibrium the distributions chosen by the consumer and the monopolist have full support on $\left[\ubar{x},\bar{x}\right]$ and $\left[\ubar{p},\bar{p}\right]$, respectively.
\end{claim}
\begin{proof}
Suppose that there is a gap in the support of the consumer's distribution, i.e., $F\left(x_{1}\right) = F\left(x_{2}\right)$ for some $x_{2} > x_{1}$. Then, $\hat{\Pi}\left(p\right) < \hat{\Pi}\left(p_{2}\right)$ for all $p \in \left(p_{1},p_{2}\right)$, where $p_{i} = x_{i} - a$. Accordingly, the firm does not charge any price $p$ in this interval, i.e., $G\left(p_{1}\right) = G^{-}\left(p_{2}\right)$.\footnote{We define $G^{-}\left(p\right) \coloneqq \sup_{w < p}G\left(w\right)$.} Then,
\[V\left(x_{1}\right) = a \left(1-G\left(p_{1}\right)\right) + \int_{\ubar{p}}^{p_{1}}\left(x_{1}-p\right)dG\left(p\right) - \kappa \left(x_{1}-\mu\right)^2 \text{ ,}\]
and
\[\begin{split}
    V\left(x_{2}\right) &= a \left(1-G\left(p_{2}\right)\right) + \int_{\ubar{p}}^{p_{2}}\left(x_{2}-p\right)dG\left(p\right) - \kappa \left(x_{2}-\mu\right)^2\\
    &= a \left(1-G\left(p_{1}\right)\right) + \int_{\ubar{p}}^{p_{1}}\left(x_{2}-p\right)dG\left(p\right) - \kappa \left(x_{2}-\mu\right)^2\\
\end{split} { .}\]
By the strict convexity of the cost function, $F$ cannot be optimal for the consumer, who could strictly benefit by having value $x' \coloneqq \lambda x_{1} + \left(1-\lambda\right)x_{2}$ for any $\lambda \in \left(0,1\right)$ in the support of $F$, rather than $x_{1}$ and $x_{2}$.

The analysis of the previous paragraph also implies that the monopolist's distribution can have no gaps in its support.
\end{proof}
The next claim establishes that the monopolist's distribution over prices may not have mass points and neither does the consumer's distribution over values except possibly at the upper bound of its support.
\begin{claim}
In any equilibrium the distribution chosen by the monopolist has no mass points and the distribution chosen by the consumer has no mass points except possibly at $\bar{x}$.
\end{claim}
\begin{proof}
Suppose for the sake of contradiction that $F$ has a mass point at some value $\hat{x} < \bar{x}$ and let $\hat{p} = \hat{x} + a$. Evidently, 
\[\hat{\Pi}\left(\hat{p} - \epsilon\right) = \left(\hat{p}-\epsilon\right)\left(1-F\left(\hat{x}-\epsilon\right)\right) > \left(\hat{p}+\epsilon\right)\left(1-F\left(\hat{x}+\epsilon\right)\right) = \hat{\Pi}\left(\hat{p} + \epsilon\right) \text{ ,}\]
for sufficiently small $\epsilon > 0$. This implies that the support of $G$ has a gap, contradicting our previous claim.

Because the consumer's distribution over values has full support, the consumer's payoff as a function of value $x$, $V\left(x\right)$, must be affine. The distribution over prices cannot have a mass point, since otherwise $V$ would have a discrete change in its slope induced by the mass point.
\end{proof}
At last we may derive the unique equilibrium. Setting $\hat{\Pi}\left(p\right) = \lambda$, we obtain
\[F\left(x\right) = 1- \frac{\lambda}{x-a} \text{ .}\]
Since $F$ has no mass points on $\left[\ubar{x},\bar{x}\right)$, $\lambda = \ubar{x}-a$. Moreover, since $\mathbb{E}_{F}\left[x\right] = \mu$, we have
\[\label{eqa1}\tag{$A2$}1 + \ln{\left\{\frac{\bar{x}-a}{\ubar{x}-a}\right\}} = \frac{\mu-a}{\ubar{x}-a} \text{ .}\]
Because $V$ is affine,
\[\label{eqa2}\tag{$A3$}a\left(1-G\left(x-a\right)\right) + \int_{\ubar{x} - a}^{x-a}\left(x-p\right)dG\left(p\right) - \kappa\left(x-\mu\right)^2 = \alpha x + \beta \text{ ,} \ \forall \ x \in \left[\ubar{x},\bar{x}\right] \text{ ,}\]
where $\alpha, \beta \in \mathbb{R}$. We may differentiate both sides and do some algebra, which yields
\[G\left(p\right) = 2 \kappa \left(a + p - \mu\right) + \alpha \text{ ,}\]
and substituting this back in to Equation \ref{eqa2}, we obtain
\[\label{eqa3}\tag{$A4$} \alpha = 2\kappa\left(\mu-\ubar{x}\right) \text{ ,}\]
and 
\[\label{eqa4}\tag{$A5$} \beta = \kappa\left(\ubar{x}^2-\mu^2\right)+a \text{ .}\]
Since $G\left(\bar{p}\right) = 1$,
\[\label{eqa5}\tag{$A6$}\alpha + 2 \kappa \left(\bar{p} + a - \mu\right) = 1 \text{ .}\]
Thus, combining Equations \ref{eqa3} and \ref{eqa5}, we get
\[\label{eqa6}\tag{$A7$}\bar{p} = \ubar{p} + \frac{1}{2\kappa} \text{ .}\]
Substituting in for $\bar{p}$ in Equation \ref{eqa1} via Equation \ref{eqa6} and using the fact that $\ubar{x} = a + \ubar{p}$, we have
\[\ln{\left\{1 + \frac{1}{2\kappa\ubar{p}}\right\}} = \frac{\mu-a}{\ubar{p}} - 1 \text{ .}\]
The firm's profit is clearly $\ubar{p}$ since if it charges that price it is purchased from with certainty, and the consumer's payoff is just $V\left(\mu\right)$.
\end{proof}

\subsection{Lemma \ref{monopprice} Proof}\label{monoppriceproof}
\begin{proof}
Let us define the function $\iota\left(\kappa, \ubar{p}, \mu\right)$ as
\[\iota\left(\ubar{p}_{M},\kappa,\mu\right) \coloneqq \ubar{p}_{M} \ln{\left\{1 + \frac{1}{2\kappa\ubar{p}_{M}}\right\}} + \ubar{p}_{M} - \mu + a \text{ .}\]
Evidently, $\iota$ is twice continuously differentiable in each of its arguments on $\mathbb{R}_{++}^{3}$.
\begin{claim}
$\iota$ has a strictly positive root if and only if $\mu > a$. This root, $\ubar{p}_{M}$, is the unique positive root and lies in the interval $\ubar{p}_{M} \in \left(0,\mu-a\right)$.
\end{claim}
\begin{proof}
Directly,
\[\iota'\left(\ubar{p}_{M}\right) = -\frac{1}{2\kappa \ubar{p}_{M}+1}+\ln\left\{\frac{1}{2\kappa \ubar{p}_{M}}+1\right\}+1, \quad \text{and} \quad \iota''\left(\ubar{p}_{M}\right) = -\frac{1}{\ubar{p}_{M}\left(2\kappa \ubar{p}_{M}+1\right)^2}\text{ .}\]
Evidently, $\iota''\left(\ubar{p}_{M}\right) < 0$ for all $\ubar{p}_{M} > 0$. We claim that $\iota$ is strictly increasing in $\ubar{p}_{M}$ for $\ubar{p}_{M} > 0$. For all such $\ubar{p}_{M}$, $\iota'\left(\ubar{p}_{M}\right)$ is strictly decreasing in $\ubar{p}_{M}$. This, plus the fact that $\lim_{\ubar{p}_{M} \to \infty}\iota'\left(\ubar{p}_{M}\right) = 1$, implies that $\iota$ is strictly increasing in $\ubar{p}_{M}$ for all $\ubar{p}_{M} > 0$.

Then, we see that as $\ubar{p}_{M} \searrow 0$, $\iota \searrow a-\mu$, which is strictly negative if and only if $\mu > a$; and as $\ubar{p}_{M} \nearrow \mu-a$, $\iota \nearrow \left(\mu-a\right) \ln\left\{1 + 1/\left(2\kappa\left(\mu-a\right)\right)\right\}$, which is strictly positive if and only if $\mu > a$. Thus, if $\mu > a$, $\iota\left(\ubar{p}_{M}\right)$ has a unique positive root at some $\ubar{p}_{M} \in \left(0, \mu-a\right)$. \end{proof}

Directly, $\iota'\left(\mu\right) = -1$ and \[\iota'\left(\kappa\right) = -\frac{\ubar{p}_{M}}{2\ubar{p}_{M}\kappa^2+\kappa} < 0 \text{ ;}\]
and so by the implicit function theorem, $\ubar{p}_{M}'\left(\mu\right) > 0$ and $\ubar{p}_{M}'\left(\kappa\right) > 0$.
\end{proof}

\subsection{Lemma \ref{monopwelfarehid} Proof}

\begin{proof}
Recall that the monopolist's profit is just the lower bound for the price, $\ubar{p}_{M}$, and so its profit is strictly increasing in $\kappa$ and $\mu$.

The consumer's payoff is $\hat{\Phi} = \kappa \left(\mu-a - \ubar{p}_{M}\right)^2 + a$. Accordingly,
\[\hat{\Phi}'\left(\mu\right) = 2\kappa\left(\mu  - a - \ubar{p}_{M}\right)\left(1 - \ubar{p}_{M}'\left(\mu\right)\right) \text{ .}\]
Then,
\[\ubar{p}_{M}'\left(\mu\right) = \frac{1}{-\frac{1}{2\kappa \ubar{p}_{M}+1}+\ln\left\{\frac{1}{2\kappa \ubar{p}_{M}}+1\right\}+1} = \frac{z + 1}{\left(z + 1\right)\ln\left\{\frac{z + 1}{z}\right\} + z} \text{ ,}\]
where $z \coloneqq 2 \ubar{p}_{M} \kappa$. Moreover, by the definition of $\ubar{p}_{M}$
\[2 \kappa \left(\mu-a - \ubar{p}_{M}\right) = z \ln{\left\{\frac{z + 1}{z}\right\}} \text{ .}\]
Consequently, $\hat{\Phi}'\left(\mu\right)$ reduces to
\[\hat{\Phi}'\left(\mu\right) =  z \ln{\left\{\frac{z + 1}{z}\right\}}\left(1-\frac{z + 1}{\left(z + 1\right)\ln\left\{\frac{z + 1}{z}\right\} + z}\right) \text{ ,}\]
which is strictly positive for all $z > 0$.

Similarly,
\[\hat{\Phi}'\left(\kappa\right) = \left(\mu - a - \ubar{p}_{M}\right)^2 - 2\kappa\left(\mu - a - \ubar{p}_{M}\right)\ubar{p}_{M}'\left(\kappa\right) \text{ ,}\]
which is positive if and only if
\[\label{ineqa7}\tag{$A8$}\kappa\left(\mu - a - \ubar{p}_{M}\right) - 2\kappa^2\ubar{p}_{M}'\left(\kappa\right) \geq 0 \text{ ,}\]
since $\mu - a > \ubar{p}_{M}$. Then,
\[\ubar{p}_{M}'\left(\kappa\right) = \frac{\frac{\ubar{p}_{M}}{2\ubar{p}_{M}\kappa^2+\kappa}}{-\frac{1}{2\kappa \ubar{p}_{M}+1}+\ln\left\{\frac{1}{2\kappa \ubar{p}_{M}}+1\right\}+1} = \left(\frac{\ubar{p}_{M}}{\kappa}\right) \frac{1}{\left(z + 1\right)\ln\left\{\frac{z + 1}{z}\right\} + z} \text{ ,}\]
where again $z \coloneqq 2 \ubar{p}_{M} \kappa$. Inequality \ref{ineqa7} is, therefore, equivalent to
\[\tag{$A9$}\label{a9a9}\frac{1}{2} \ln{\left\{\frac{z + 1}{z}\right\}} - \frac{1}{\left(z + 1\right)\ln\left\{\frac{z + 1}{z}\right\} + z} \geq 0 \text{ ,}\]
which is equivalent to requiring that $\kappa \ubar{p}_{M} \lesssim .337$.
\end{proof}

\subsection{Proposition \ref{auxhelp} Proof}
\begin{proof}
First, we show that $\ubar{p}_{M}$ is strictly concave in $\kappa$ and $\ubar{p}$ is strictly convex in $\kappa$. 
\begin{claim}
$\ubar{p}_{M}$ is strictly concave in $\kappa$.
\end{claim}
\begin{proof}
By the implicit function theorem (and dropping the subscript $M$),
\[\ubar{p}''\left(\kappa\right) = \frac{-\iota_{\ubar{p}}^2\iota_{\kappa\kappa} + 2\iota_{\ubar{p}}\iota_{\kappa}\iota_{\kappa \ubar{p}}-\iota_{\kappa}^2\iota_{\ubar{p}\ubar{p}}}{\iota_{\ubar{p}}^3} \text{ .}\]
Since $\iota_{\ubar{p}} > 0$ this is strictly negative if and only if 
\[\label{ineq8}\tag{$A10$}-\iota_{\ubar{p}}^2\iota_{\kappa\kappa} + 2\iota_{\ubar{p}}\iota_{\kappa}\iota_{\kappa \ubar{p}}-\iota_{\kappa}^2\iota_{\ubar{p}\ubar{p}} < 0 \text{ .}\]
Let us go through each term one by one. First,
\[-\iota_{\ubar{p}}^2\iota_{\kappa\kappa} = -\left(\frac{z}{z+1}+\ln\left\{\frac{1}{z}+1\right\}\right)^2 \frac{\ubar{p}\left(2 z+1\right)}{\kappa^2\left(z+1\right)^2} \text{ ,}\]
where $z \coloneqq 2\kappa \ubar{p}$. Second,
\[2\iota_{\ubar{p}}\iota_{\kappa}\iota_{\kappa \ubar{p}} = 2\left(\frac{z}{z+1}+\ln\left\{\frac{1}{z}+1\right\}\right)\frac{\ubar{p}}{\kappa^2}\left(\frac{1}{\left(z+1\right)^3}\right) \text{ .}\]
Third,
\[-\iota_{\kappa}^2\iota_{\ubar{p}\ubar{p}} =  \frac{\ubar{p}}{\kappa^2}\left(\frac{1}{\left(z+1\right)^4}\right)\]
Summing these together (and dividing out by positive terms), Inequality \ref{ineqa8} holds if and only if
\[\label{ineqa11}\tag{$A11$}2\ln\left(\frac{1}{z}+1\right)\left(1-2z\right)+\frac{2z+1}{z+1}\left(1-z\right)-\left(\ln\left(\frac{1}{z}+1\right)\right)^{2}\left(2z+1\right) < 0 \text{ ,}\]
which holds for all strictly positive $z$.
\end{proof}
\begin{claim}
$\ubar{p}$ is strictly convex in $\kappa$.
\end{claim}
\begin{proof}
Let us define the function $\varphi\left(\kappa, \ubar{p}, c\right)$ as
\[\varphi\left(\kappa, \ubar{p}, c\right) \coloneqq \ubar{p}\ln{\left\{1 + \frac{1}{2\kappa\ubar{p}}\right\}} - \sqrt{\frac{c}{\kappa}} \text{ .}\]
Evidently, $\varphi$ is twice continuously differentiable in each of its variables on $\mathbb{R}_{++}^{3}$. 

We shall break protocol slightly in this proof and use results that are derived (not much) later on in the paper (that $\varphi$ is strictly increasing in $\ubar{p}$ for $\ubar{p} > 0$, which we establish in Section \ref{costathid1proof}). Mirroring the proof of the previous claim, we need to show that \[\label{ineqa9}\tag{$A12$}-\varphi_{\ubar{p}}^2\varphi_{\kappa\kappa} + 2\varphi_{\ubar{p}}\varphi_{\kappa}\varphi_{\kappa \ubar{p}}-\varphi_{\kappa}^2\varphi_{\ubar{p}\ubar{p}} > 0 \text{ .}\]
Again, let us go through each term one by one. First, 
\[\begin{split}
    -\varphi_{\ubar{p}}^2\varphi_{\kappa\kappa} &= -\left(\ln{\left\{1 + \frac{1}{x}\right\}} - \frac{1}{1+x}\right)^2\left(\frac{x\left(2x+1\right)}{2 \left(1+x\right)^2\kappa^3} - \frac{3 \sqrt{c\kappa}}{4\kappa^3}\right)\\
    &= -\left(\ln{\left\{1 + \frac{1}{x}\right\}} - \frac{1}{1+x}\right)^2\left(\frac{x\left(2x+1\right)}{2 \left(1+x\right)^2\kappa^3}-\frac{3x\ln{\left\{1 + \frac{1}{x}\right\}}}{8\kappa^3}\right)
\end{split} \text{ ,}\]
where $x \coloneqq 2\kappa \ubar{p}$ and since we must have \[\sqrt{c\kappa}= \frac{x}{2}\ln{\left\{1+\frac{1}{x}\right\}} \text{ .}\] Second,
\[\begin{split}
    2\varphi_{\ubar{p}}\varphi_{\kappa}\varphi_{\kappa \ubar{p}} &= -2\left(\ln{\left\{1 + \frac{1}{x}\right\}} - \frac{1}{1+x}\right)\left(\frac{\sqrt{c\kappa}}{2\kappa^2}-\frac{x}{2\left(x+1\right)\kappa^2}\right)\left(\frac{1}{\kappa\left(x+1\right)^2}\right)\\
    &= -2\left(\ln{\left\{1 + \frac{1}{x}\right\}} - \frac{1}{1+x}\right)\left(\frac{x\ln{\left\{1 + \frac{1}{x}\right\}}}{4\kappa^2}-\frac{x}{2\left(x+1\right)\kappa^2}\right)\left(\frac{1}{\kappa\left(x+1\right)^2}\right)\\
\end{split} \text{.}\]
Third,\[\begin{split}
    -\varphi_{\kappa}^2\varphi_{\ubar{p}\ubar{p}} &= \left(\frac{x\ln{\left\{1 + \frac{1}{x}\right\}}}{4\kappa^2}-\frac{x}{2\left(x+1\right)\kappa^2}\right)^2\left(\frac{1}{\ubar{p}}\right)\frac{1}{\left(1+x\right)^2}
\end{split} \text{.}\]
Summing these together (and dividing out by positive terms), Inequality \ref{ineqa9} holds if and only if
\[3\ln^{3}\left(\frac{1}{x}+1\right)\left(x+1\right)^{3}+\ln^{2}\left(\frac{1}{x}+1\right)\left(x+1\right)\left(-14x-13\right)+19\ln\left(\frac{1}{x}+1\right)\left(x+1\right)-8 > 0 \text{ ,}\]
which holds for all strictly positive $x$.
\end{proof}

Next, we establish that $\ubar{p}_{M}$ ranges from $0$ to some finite positive number and that $\ubar{p}$ is U-shaped in $\kappa$ and explodes at $\kappa = 0$ and $\kappa = 1/\left(4c\right)$.
\begin{claim}
$\infty > \lim_{\kappa \nearrow 1/\left(4c\right)} \ubar{p}_{M} > 0$ and $\lim_{\kappa \searrow 0} \ubar{p}_{M} = 0$.
\end{claim}
\begin{proof}
The first statement in the claim follows immediately from the analysis in Section \ref{monoppriceproof}. Suppose for the sake of contradiction that $\ubar{p}_{M} > L$ for some constant $L > 0$ for all $\kappa > 0$. Thus, 
\[\lim_{\kappa \searrow 0} \ubar{p}_{M} \ln{\left\{1 + \frac{1}{2\kappa\ubar{p}_{M}}\right\}} + \ubar{p}_{M} - \mu \geq \lim_{\kappa \searrow 0} L \ln{\left\{1 + \frac{1}{2\kappa L}\right\}} + L - \mu = \infty \text{ .}\]
Since $\iota$ is continuous and strictly decreasing in $\kappa$, this means that there is some $\hat{\kappa} > 0$ such that for all $0 < \kappa < \hat{\kappa}$, $\iota\left(\kappa\right) > 0$, a contradiction.
\end{proof}

\begin{claim}
$\lim_{\kappa \nearrow 1/\left(4c\right)}\ubar{p} = \lim_{\kappa \searrow 0}\ubar{p} = \infty$.
\end{claim}
\begin{proof}
As is shown in Section \ref{costathid1proof}, $\ubar{p}$ is strictly increasing in $\kappa$ for all $\kappa \in \left(\tilde{\kappa}, 1/\left(4c\right)\right)$, where $\tilde{\kappa} \in \left(0, 1/\left(4c\right)\right)$. Suppose for the sake of contradiction that $\ubar{p}$ is bounded above by $K < \infty$ for all $\kappa \in \left(\tilde{\kappa}, 1/\left(4c\right)\right)$. We have
\[\lim_{\kappa \nearrow \frac{1}{4c}}\ubar{p}\ln{\left\{1 + \frac{1}{2\kappa\ubar{p}}\right\}} - \sqrt{\frac{c}{\kappa}} \leq \lim_{\kappa \nearrow \frac{1}{4c}}K\ln{\left\{1 + \frac{1}{2\kappa K}\right\}} - \sqrt{\frac{c}{\kappa}} = K\ln{\left\{1 + \frac{2c}{K}\right\}} - 2c < 0 \text{ ,}\]
a contradiction.

Similarly, $\ubar{p}$ is strictly decreasing in $\kappa$ for all $\kappa \in \left(0, \tilde{\kappa}\right)$. Suppose for the sake of contradiction that $\ubar{p}$ is bounded above by $M < \infty$ for all $\kappa \in \left(0, \tilde{\kappa}\right)$. We have
\[\lim_{\kappa \searrow 0} \frac{\ubar{p}\ln{\left\{1 + \frac{1}{2\kappa\ubar{p}}\right\}}}{\sqrt{\frac{c}{\kappa}}} \leq \lim_{\kappa \searrow 0} \frac{M\ln{\left\{1 + \frac{1}{2\kappa M}\right\}}}{\sqrt{\frac{c}{\kappa}}} = 0 < 1 \text{ ,}\]
a contradiction.
\end{proof}

Next, (setting $a = 0$)
\[\iota\left(\mu - \sqrt{\frac{c}{\kappa}}\right) = \left(\mu - \sqrt{\frac{c}{\kappa}}\right) \ln{\left\{1 + \frac{1}{2\kappa\left(\mu - \sqrt{\frac{c}{\kappa}}\right)}\right\}} - \sqrt{\frac{c}{\kappa}} \eqqcolon \rho\left(\mu\right) \text{ ,}\]
which is equal to $0$ for $\kappa > 0$ if and only if
\[x\left(\mu\right) \ln{\left(1+\frac{1}{x\left(\mu\right)}\right)} = 2 \sqrt{c \kappa} \text{ ,}\]
where $x\left(\mu\right) \coloneqq 2 \kappa \mu - 2 \sqrt{c\kappa}$. Then, it is easy to see that the left-hand side of this equation is strictly increasing and strictly concave in $\mu$ and converges to $1$ as $\mu$ goes to $\infty$. On the other hand, the right-hand side is a constant (for $c$ and $\kappa$ fixed) and is strictly less than $1$, since $\kappa < 1/\left(4c\right)$. Thus, for all (allowed) $\kappa$ and $c$, there exists a $\mu$ such that the equilibrium involves active search. Furthermore, if $\mu \leq \ (<) \sqrt{c/\kappa}$, $\ubar{p}_{M}$ is obviously (strictly) greater than $\mu - \sqrt{c/\kappa}$.

Finally, since $\ubar{p}_{M}$ is continuous and strictly increasing in $\mu$ for $\mu > 0$ and ranges from $0$ to some positive value, $\ubar{p}$ is U-shaped and explodes at its endpoints, and $\ubar{p}_{M}$ is strictly concave in $\kappa$ and $\ubar{p}$ is strictly convex in $\kappa$, the result follows.
\end{proof}

\subsection{Proposition \ref{costathid1} Proof}\label{costathid1proof}
\begin{proof}
Let us begin with the following claim.
\begin{claim}\label{priceexistence}
$\varphi\left(\ubar{p}\right)$ has a real root at some $\ubar{p} > 0$ if and only if $c < 1/\left(4\kappa\right)$. This root is unique. 
\end{claim}

\begin{proof}
First, we argue that $\varphi$ cannot have a root at some $\ubar{p} < 0$. Directly,
\[\varphi'\left(\ubar{p}\right) = \ln{\left\{1 + \frac{1}{2\kappa\ubar{p}}\right\}} - \frac{1}{1+2\kappa\ubar{p}}, \quad \text{and} \quad \varphi''\left(\ubar{p}\right) = -\frac{1}{\ubar{p}\left(2\kappa\ubar{p}+1\right)^2} \text{ .}\]
Evidently, $\varphi''\left(\ubar{p}\right) > 0$ for all $\ubar{p} < 0$. Consequently, since $\varphi'\left(\ubar{p}\right)$ is strictly increasing in $\ubar{p}$ for all $\ubar{p} \in \left(-\infty,0\right)$ and $\lim_{\ubar{p} \to -\infty}\varphi'\left(\ubar{p}\right) = 0$, $\varphi$ is strictly increasing in $\ubar{p}$ for all $\ubar{p} \in \left(-\infty,0\right)$. Since $\lim_{\ubar{p} \nearrow 0}\iota\left(\ubar{p}_{M}\right) = -\sqrt{c/\kappa} < 0$, we conclude that there is no root.

Second, $\lim_{\ubar{p} \searrow 0}\varphi\left(\ubar{p}\right) = -\sqrt{c/\kappa}$, thereby eliminating $0$ as a candidate root.

Third, we establish that $\varphi$ is strictly increasing in $\ubar{p}$ for strictly positive $\ubar{p}$. Since $\varphi'\left(\ubar{p}\right)$ is strictly decreasing in $\ubar{p}$ and $\lim_{\ubar{p} \to \infty}\varphi'\left(\ubar{p}\right) = 0$, $\varphi$ is strictly increasing in $\ubar{p}$ for all $\ubar{p}$. 

Fourth, as $\ubar{p} \searrow 0$, $\varphi \searrow -\sqrt{c/\kappa} < 0$; and as $\ubar{p} \nearrow + \infty$, $\varphi \nearrow 1/\left(2\kappa\right) - \sqrt{c/\kappa}$. Thus, $\varphi\left(\ubar{p}\right)$ has a positive real root if and only if $c < 1/\left(4\kappa\right)$. It is unique.
\end{proof} 

Evidently, $\varphi'\left(c\right) < 0$. Then, by the implicit function theorem,
\[\ubar{p}'\left(c\right) = -\frac{\varphi'\left(c\right)}{\varphi'\left(\ubar{p}\right)} > 0 \text{,}\]
where we used the result from Claim \ref{priceexistence} that $\varphi'\left(\ubar{p}\right) > 0$.

Next, 
\[\varphi'\left(\kappa\right) = \frac{\sqrt{c\kappa}}{2\kappa^2}-\frac{1}{2\left(\frac{1}{2\ubar{p}\kappa}+1\right)\kappa^2} \text{ ,}\]
which is positive if and only if
\[\ubar{p} \leq \frac{\sqrt{c}}{2\sqrt{\kappa}-2\kappa\sqrt{c}} \text{ .}\]
Then,
\[\begin{split}
    \varphi\left(\frac{\sqrt{c}}{2\sqrt{\kappa}-2\kappa\sqrt{c}}\right) \geq 0 \quad &\Leftrightarrow \quad \frac{\sqrt{c}}{2\sqrt{\kappa}-2\kappa\sqrt{c}}\ln{\left\{1 + \frac{1}{2\kappa\frac{\sqrt{c}}{2\sqrt{\kappa}-2\kappa\sqrt{c}}}\right\}} - \sqrt{\frac{c}{\kappa}} \geq 0\\
    &\Leftrightarrow \quad -\ln{\left\{\sqrt{c \kappa}\right\}} - 2+2\sqrt{c\kappa} \geq 0
\end{split} \text{ .}\]
By the implicit function theorem, $\ubar{p}$ is decreasing in $\kappa$ if and only if $\varphi'\left(\kappa\right) \geq 0$, if and only if $h\left(c,\kappa\right) \coloneqq \ln{\left\{\sqrt{c \kappa}\right\}} + 2 - 2\sqrt{c\kappa} \leq 0$. Note that $h$ has a unique root at $\tilde{\kappa} \in \left(0, 1/\left(4c\right)\right)$.
\end{proof}

\subsection{Lemma \ref{payoffkappa} Proof}
\begin{proof}
Let us begin with the result pertaining to the explicit search cost $c$. Directly, 
\[\hat{\Phi}'\left(c\right) = 1 - \ubar{p}'\left(c\right) - \frac{1}{2\sqrt{c\kappa}} < 0 \text{ ,}\]
 since $\ubar{p}'\left(c\right) > 0$ and $c < 1/\left(4\kappa\right)$. Next, 
 \[\hat{\Phi}'\left(\kappa\right) = - \ubar{p}'\left(\kappa\right) + \frac{1}{2\kappa}\sqrt{\frac{c}{\kappa}} \geq 0 \quad \Leftrightarrow \quad \frac{\varphi'\left(\kappa\right)}{\varphi'\left(\ubar{p}\right)} \geq  - \frac{1}{2\kappa}\sqrt{\frac{c}{\kappa}} \text{ .}\]
 Substituting in for $\varphi'\left(\kappa\right)$ and $\varphi'\left(\ubar{p}\right)$, this holds if and only if
 \[\frac{\frac{\sqrt{c\kappa}}{2\kappa^2}-\frac{1}{2\left(\frac{1}{2\ubar{p}\kappa}+1\right)\kappa^2}}{\ln{\left\{1 + \frac{1}{2\kappa\ubar{p}}\right\}} - \frac{1}{1+2\kappa\ubar{p}}} \geq  - \frac{1}{2\kappa}\sqrt{\frac{c}{\kappa}} \quad \Leftrightarrow \quad 2\ubar{p}\kappa\left(\sqrt{c\kappa}-1\right)+\sqrt{c\kappa}\ln\left(1+\frac{1}{2\kappa\ubar{p}}\right)\left(1+2\ubar{p}\kappa\right) \geq 0 \text{ .}\]
 Enacting a change of variable with $x \coloneqq \kappa \ubar{p}$ and substituting in using $\varphi = 0$, the latter inequality becomes
 \[w\left(x\right) \coloneqq \left(1+\frac{1}{2x}\right)\ln{\left\{1+\frac{1}{2x}\right\}} + 1 - \frac{1}{x\ln{\left\{1+\frac{1}{2x}\right\}}} \geq 0 \text{ .}\]
 Note that $x$ is strictly increasing in $\kappa$ for all $\kappa \in \left(\tilde{\kappa},1/\left(4c\right)\right)$. Define $\tilde{x} \coloneqq \tilde{\kappa} \ubar{p}\left(\tilde{\kappa}\right)$ and observe that by construction, $w\left(\tilde{x}\right) > 0$. Moreover, $w$ is strictly decreasing in $x$ and $\lim_{x \to \infty}w\left(x\right) = -1$. Appealing to the intermediate value theorem, we deduce the result. \end{proof}

\subsection{Proposition \ref{monopcompare} Proof}\label{monopcompareproof}
\begin{proof}Recall
\[\iota\left(\ubar{p},\kappa,\mu\right) \coloneqq \ubar{p} \ln{\left\{1 + \frac{1}{2\kappa\ubar{p}}\right\}} + \ubar{p} - \mu \text{ ,}\]
(where we have set $a = 0$), which is strictly increasing in $\ubar{p}$. 

First, let $\mu \geq 3/\left(4\kappa\right)$. We have
\[\iota\left(p\right) < \iota\left(\mu-\frac{1}{2\kappa}\right) = \left(\mu-\frac{1}{2\kappa}\right) \ln{\left\{1 + \frac{1}{2\kappa\left(\mu-\frac{1}{2\kappa}\right)}\right\}} + \left(\mu-\frac{1}{2\kappa}\right) - \mu \text{ ,}\]
for all $p < \mu - 1/\left(2\kappa\right)$. The right-hand side of this expression is weakly positive if and only if
\[\left(z - 1\right)\ln{\left\{\frac{z}{z - 1}\right\}} - 1 \geq 0 \text{ , }\]
where $z \coloneqq 2\kappa \mu$, which never holds. Thus, $\ubar{p} \geq \mu - 1/\left(2\kappa\right)$. We also have
\[\iota\left(p\right) > \iota\left(\mu-\frac{1}{4\kappa}\right) = \left(\mu-\frac{1}{4\kappa}\right) \ln{\left\{1 + \frac{1}{2\kappa\left(\mu-\frac{1}{4\kappa}\right)}\right\}} + \left(\mu-\frac{1}{4\kappa}\right) - \mu \text{ ,}\]
for all $p > \mu - 1/\left(4\kappa\right)$. The right-hand side of this expression is weakly positive if and only if \[\tag{$A13$}\label{ineqa8}\left(z - 1\right)\ln{\left\{\frac{z+1}{z - 1}\right\}} - 1 \geq 0 \text{ , }\]
where $z \coloneqq 4\kappa \mu$. Since $\mu \geq 3/\left(4\kappa\right)$, $z \geq 3$ and so Inequality \ref{ineqa8} holds. Consequently, $\ubar{p} \leq \mu - 1/\left(4\kappa\right)$.

Second, let $3/\left(4\kappa\right) \geq \mu > 0$. We have
\[\iota\left(p\right) < \iota\left(\frac{\mu}{2}-\frac{1}{8\kappa}\right) = \left(\frac{\mu}{2}-\frac{1}{8\kappa}\right) \ln{\left\{1 + \frac{1}{2\kappa\left(\frac{\mu}{2}-\frac{1}{8\kappa}\right)}\right\}} + \left(\frac{\mu}{2}-\frac{1}{8\kappa}\right) - \mu \text{ ,}\]
for all $p < \mu/2 - 1/\left(8\kappa\right)$. The right-hand side of this expression is weakly positive if and only if
\[\left(z - 1\right)\ln{\left\{\frac{z+3}{z - 1}\right\}} - 1 - z \geq 0 \text{ , }\]
where $z \coloneqq 4\kappa \mu$, which never holds. Thus, $\ubar{p} \geq \mu/2 - 1/\left(8\kappa\right)$. Finally, we also have
\[\iota\left(p\right) > \iota\left(\frac{\left(1+4\kappa\mu\right)^2}{32\kappa}\right) = \frac{\left(1+4\kappa\mu\right)^2}{32\kappa}\ln{\left\{1 + \frac{1}{2\kappa\frac{\left(1+4\kappa\mu\right)^2}{32\kappa}}\right\}} + \frac{\left(1+4\kappa\mu\right)^2}{32\kappa} - \mu \text{ ,}\]
for all $p > \left(1+4\kappa\mu\right)^2/\left(32\kappa\right)$. The right-hand side of this expression is weakly positive if and only if \[\left(1 + z\right)^2\ln{\left\{\frac{(1+z)^2+16}{(1+z)^2}\right\}} + \left(1 + z\right)^2 - 8 z \geq 0 \text{ , }\]
where $z \coloneqq 4\kappa \mu$. This clearly holds. Consequently, $\ubar{p} < \left(1+4\kappa\mu\right)^2/\left(32\kappa\right)$.
\end{proof}

\subsection{Proposition \ref{oligopco} Proof}

\begin{proof}
First, we establish an auxiliary result.
 \begin{claim}\label{highmix}
 If $\mu \geq \ubar{p} + \sqrt{c/\kappa}$, then  $\ubar{p} > \sqrt{c/\kappa} - 1/\left(4\kappa\right)$.
 \end{claim}
 \begin{proof}
 If $c \leq 1/\left(16\kappa\right)$, the result is immediate, so let $c > 1/\left(16\kappa\right)$. Since $\varphi$ is strictly increasing in $\ubar{p}$ for all $\ubar{p} > 0$, we have
 \[\varphi\left(\ubar{p}\right) \leq \varphi\left(\sqrt{\frac{c}{\kappa}} - \frac{1}{4\kappa}\right) = \left(\sqrt{\frac{c}{\kappa}} - \frac{1}{4\kappa}\right)\ln{\left\{\frac{4\sqrt{c\kappa}+1}{4\sqrt{c\kappa}-1}\right\}} - \sqrt{\frac{c}{\kappa}} \text{ ,}\]
 for all $\ubar{p} \leq \mu - \sqrt{c/\kappa}$. The right hand side of this expression is weakly positive if and only if
 \[f\left(x\right) \coloneqq \left(x-1\right)\ln{\left\{\frac{x+1}{x-1}\right\}} - x \geq 0 \text{ ,}\]
 where $x \coloneqq 4\sqrt{c\kappa}$. Note that $x$ takes values in the interval $\left(1,2\right)$, and it is easy to see that $f < 0$ for all such $x$. Therefore,  $\ubar{p} > \sqrt{c/\kappa} - 1/\left(4\kappa\right)$.
 \end{proof}
 Second, we establish the result for high $\mu$.
 \begin{claim}
 If $\mu \geq \ubar{p} + \sqrt{c/\kappa}$, then  $\Phi_{l} > \hat{\Phi}$.
 \end{claim}
 \begin{proof}
 \[\Phi_{l} > \hat{\Phi} \quad \Leftrightarrow \quad \mu + \frac{1}{4\kappa} + c - 2 \sqrt{\frac{c}{\kappa}} > \mu - \ubar{p} + c - \sqrt{\frac{c}{\kappa}} \quad \Leftrightarrow \quad \ubar{p} > \sqrt{\frac{c}{\kappa}} - \frac{1}{4\kappa} \text{ ,}\]
which follows from Claim \ref{highmix}.
 \end{proof}
 Third, we establish the result for moderate $\mu$.
 \begin{claim}
If $\ubar{p} + \sqrt{c/\kappa} > \mu \geq 2 \sqrt{c/\kappa} - 1/\left(4\kappa\right)$, then $\Phi_{l} > \hat{\Phi}_{M}$.
 \end{claim}
 \begin{proof}
 Using the conditions of the claim,
 \[\Phi = \mu + \frac{1}{4\kappa} + c - 2 \sqrt{\frac{c}{\kappa}} \geq c > \kappa \left(\mu - \ubar{p}\right)^2 = \hat{\Phi}_{M} \text{ .}\]
 \end{proof}
 Finally, if $\ubar{p} + \sqrt{c/\kappa} > \mu$ and $2 \sqrt{c/\kappa} - 1/\left(4\kappa\right) > \mu$, the result is derived in the proof of Proposition \ref{monopcompare}.
\end{proof}

\subsection{Proposition \ref{olicopco2} Proof}

\begin{proof}
First, we establish the result for high $\mu$.
 \begin{claim}
 If $\mu \geq \ubar{p} + \sqrt{c/\kappa}$, then  $\ubar{p} > 2c$.
 \end{claim}
 \begin{proof}
 By Claim \ref{highmix},if $\mu \geq \ubar{p} + \sqrt{c/\kappa}$, then  $\ubar{p} > \sqrt{c/\kappa} - 1/\left(4\kappa\right)$, i.e., if there is active search when prices are hidden, there must be active search when prices are observable. Since $\varphi$ is strictly increasing in $\ubar{p}$ for all $\ubar{p} > 0$, we have
 \[\varphi\left(\ubar{p}\right) \leq \varphi\left(2 c\right) = 2 c\ln{\left\{\frac{4\kappa c+1}{4\kappa c}\right\}} - \sqrt{\frac{c}{\kappa}} \text{ ,}\]
 for all $\ubar{p} \leq 2c$. The right hand side of this expression is negative if and only if
 \[x \ln{\left\{\frac{x+1}{x}\right\}} - \sqrt{x} < 0 \text{ ,}\]
 where $x \coloneqq 4 \kappa c$. This clearly holds for all $x \in \left(0,1\right)$.
 \end{proof}
 Thus, if the mean is sufficiently high so that there is active search when prices are hidden, firms prefer that prices be hidden. Second, we establish that for medium $\mu$, firms prefer hidden before learning if and only if $\mu$ is sufficiently high.
 \begin{lemma}
 Let $\mu > 0$. There exist $\kappa$, $\mu$ and $c$ that generate $\ubar{p} + \sqrt{c/\kappa} > \mu \geq 2 \sqrt{c/\kappa} - 1/\left(4\kappa\right)$, such that $\ubar{p} > 2c$, $\ubar{p} = 2c$, or $\ubar{p} < 2c$.
 \end{lemma}
  \begin{proof}
  Note, that the condition for this result is that there is active search if and only if prices are observable. Since $\iota$ is strictly decreasing in $\mu$ it suffices to show that there exist values of the parameters such that $\ubar{p} = 2c$, since at the two endpoints ($\ubar{p} + \sqrt{c/\kappa}$ and $2\sqrt{c/\kappa} - 1/\left(4\kappa\right)$) firms strictly prefer active or hidden search, respectively. Directly, we have
 \[\iota\left(2 c\right) = 2c \ln{\left\{\frac{4\kappa c+1}{4\kappa c}\right\}} + 2 c - \mu \text{ .}\]
 This must equal $0$ at equilibrium, i.e., 
 \[\mu = 2c \ln{\left\{\frac{4\kappa c+1}{4\kappa c}\right\}} + 2 c\text{ .}\]
 Consequently, we must have
 \[2 c + \sqrt{\frac{c}{\kappa}} > 2c \ln{\left\{\frac{4\kappa c+1}{4\kappa c}\right\}} + 2 c \geq 2 \sqrt{\frac{c}{\kappa}} - \frac{1}{4\kappa} \text{ ,}\]
 which holds if and only if
 \[\sqrt{x} > x\ln\left\{\frac{x+1}{x}\right\} \geq 2\sqrt{x} - x-\frac{1}{2} \text{ ,}\]
 where $x \coloneqq 4 \kappa c$. This holds for all $x \in \left(0,1\right)$. 
 \end{proof}
 Third, if $\ubar{p} + \sqrt{c/\kappa} > \mu$ and $2 \sqrt{c/\kappa} - 1/\left(4\kappa\right) > \mu$, then the proof of Proposition \ref{monopcompare} implies that $\ubar{p} < \left(1+4\kappa\mu\right)^2/\left(32\kappa\right)$. That is, if there is no active search for either timing regime, firms prefer that prices be observable.
 
 Finally, if $0 \geq \mu > -1/\left(4\kappa\right)$, the firms obviously prefer that prices are not hidden since otherwise there is no trade.
\end{proof}

\end{document}